\documentclass[sigconf]{acmart}

\usepackage[noend]{algorithmic}
\usepackage{subfigure}
\usepackage{multirow}
\usepackage{color,soul}
\usepackage{array}
\usepackage{balance}
\usepackage{paralist}
\usepackage{enumitem}
\usepackage{url}
\usepackage{xspace}

\newenvironment{citemize}
{\begin{compactitem}}
{\end{compactitem}}

\newenvironment{cenumerate}
{\begin{compactenum}}
{\end{compactenum}}

\def\e#1{\emph{#1}}
\newcommand{\eat}[1]{}
\newcommand{\eqdef}{\overset{\mathrm{def}}{=\joinrel=}}

\newenvironment{citedtheorem}[1]
{\begin{theorem}{\it\e{(#1)}}\,\,}
{\end{theorem}}

\def\nodes{\mathsf{V}}
\def\edges{\mathsf{E}}
\def\set#1{\mathord{\{#1\}}}
\def\T{\mathcal{T}}
\def\angs#1{\mathord{\langle#1\rangle}}

\newenvironment{repeatresult}[2]
{\vskip0.5em\par\textsc{#1} #2.\em}
{\vskip1em}

\newenvironment{repproposition}[1]{\begin{repeatresult}{Proposition}{#1}}{\end{repeatresult}}

\newenvironment{replemma}[1]{\begin{repeatresult}{Lemma}{#1}}{\end{repeatresult}}

\newcommand{\algname}[1]{{\sf #1}}
\def\asn{\mathbin{{:}{=}}}
\def\myrulewidth{3.20in}
\def\therule{\rule{\myrulewidth}{0.2pt}}

\newenvironment{algseries}[2]
{\centering\begin{figure}[#1]\begin{center}\def\thecaption{\caption{#2}}
\begin{tabular}{p{\myrulewidth}}\therule\end{tabular}\vskip0.2em}
{\thecaption\end{center}\end{figure}}

\newenvironment{insidealg}[2]
{\normalsize\begin{insidecode}{#1}{#2}{Algorithm}}
{\end{insidecode}}

\newenvironment{insidecode}[3]
{
\begin{tabular}{p{\myrulewidth}}
\multicolumn{1}{c}{\rule{0mm}{3mm}{\bf #3} $\algname{#1}(\mbox{#2})$\vspace{-0.6em}}\\
\therule\vskip-0.8em\therule
\vspace{-1em}
\begin{algorithmic}[1]}
{\end{algorithmic}
\vskip-0.3em\therule
\end{tabular}}

\newcommand{\induced}[2]{#1[#2]}
\newcommand{\HByPMC}[2]{H_{#1}(#2)}

\def\maxcliques{\mathit{MaxClq}}
\def\minseps{\mathit{MinSep}}
\def\bags{\mathrm{bags}}
\def\pmcs{\mathit{PMC}}

\def\fblocks{\mathit{Blck}}
\def\width{\mathrm{width}}
\def\fillin{\mathrm{fill\mbox{-}in}}

\def\K{\mathcal{K}}
\def\R{\mathcal{R}}
\def\P{\mathsf{P}}
\def\G{\mathbf{G}}

\newcommand{\noamdone}[1]{}

\theoremstyle{plain}
\newtheorem{theorem}{Theorem}[section]
\newtheorem{proposition}[theorem]{Proposition}
\newtheorem{lemma}[theorem]{Lemma}
\newtheorem{corollary}[theorem]{Corollary}

\def\partitle#1{\subsubsection*{\noindent\underline{#1}}}

\begin{document}

\title{Ranked Enumeration of Minimal Triangulations}
\subtitle{Tracks: Exploration, Experimental}

\author{Noam Ravid}
\affiliation{
\institution{Technion}
\city{Haifa}
\country{Israel}
\postcode{32000}}
\email{noamrvd@cs.technion.ac.il}

\author{Dori Medini}
\affiliation{
\institution{Technion}
\city{Haifa}
\country{Israel}
\postcode{32000}}
\email{dorimedi@campus.technion.ac.il}

\author{Benny Kimelfeld}
\affiliation{
\institution{Technion}
\city{Haifa}
\country{Israel}
\postcode{32000}}
\email{bennyk@cs.technion.ac.il}


\begin{abstract}
  A tree decomposition of a graph facilitates computations by grouping
  vertices into bags that are interconnected in an acyclic structure;
  hence their importance in a plethora of problems such as query
  evaluation over databases and inference over probabilistic graphical
  models. The relative benefit from different tree decompositions is
  measured by diverse (sometime complex) cost functions that vary from
  one application to another. For generic cost functions like width
  and fill-in, an optimal tree decomposition can be efficiently
  computed in some cases, notably when the number of minimal
  separators is bounded by a polynomial (due to Bouchitte and
  Todinca); we refer to this assumption as ``poly-MS.'' To cover the
  variety of cost functions in need, it has recently been proposed to
  devise algorithms for enumerating many decomposition candidates for
  applications to choose from using specialized, or even
  machine-learned, cost functions.
 
  We explore the ability to produce a large collection of ``high
  quality'' tree decompositions. We present the first algorithm for
  ranked enumeration of the proper (non-redundant) tree
  decompositions, or equivalently minimal triangulations, under a wide
  class of cost functions that substantially generalizes the above
  generic ones.  On the theoretical side, we establish the guarantee
  of polynomial delay if poly-MS is assumed, or if we are interested
  in tree decompositions of a width bounded by a constant. We describe
  an experimental evaluation on graphs of various domains (including
  join queries, Bayesian networks, treewidth benchmarks and random),
  and explore both the applicability of the poly-MS assumption and the
  performance of our algorithm relative to the state of the art.
\end{abstract}

\maketitle

\section{Introduction}
A \e{tree decomposition} of a graph $G$ is a tree $\T$ such that each
vertex of $\T$ is associated with a \e{bag} of vertices of $G$, every
edge of $G$ appears in at least one bag, and every vertex of $G$
occurs in a connected subtree of $\T$. Tree decompositions are useful
in common scenarios where problems are intractable on general
structures, yet tractable on acyclic ones.
A \e{beneficial} tree decomposition allows for efficient computation.
This benefit is typically estimated by a cost function, most popular
being the \e{width}---the cardinality of the largest bag (minus one),
and the \e{fill in}--the number of missing edges among bag neighbors.
The generalization to hypergraphs, \e{generalized hypertree
  decomposition}, is a tree decomposition of the \e{primal} graph
(consisting of an edge between hyperedge neighbors) along with a cover
of each bag by hyperedges, giving rise to specialized
costs~\cite{DBLP:conf/wg/GottlobGMSS05} such as (\e{generalized})
\e{hypertree
  width}~\cite{GOTTLOB2002579,DBLP:journals/jacm/GottlobMS09}, and
\e{fractional hypertree width}~\cite{DBLP:journals/talg/Marx10}.  The
applications of (ordinary and generalized) tree decompositions include
optimization of join queries in
databases~\cite{DBLP:conf/sigmod/TuR15,DBLP:journals/jair/GottlobGS05},
solvers for constraint satisfaction
problems~\cite{DBLP:journals/jcss/KolaitisV00}, RNA analysis in
bioinformatics~\cite{DBLP:conf/wabi/ZhaoMC06}, computation of Nash
equilibria in game theory~\cite{DBLP:journals/jair/GottlobGS05},
inference in probabilistic graphical models~\cite{LauSpi-JRS88}, and
weighted model counting~\cite{DBLP:conf/sum/KenigG15}.

Computing an optimal tree decomposition is NP-hard for the classic
cost measures as the aforementioned ones.  Therefore, heuristic
algorithms are often
used~\cite{Berry:2002:MCS:647683.732496,BERRY200633}.  But even
regardless of the computational hardness, applications often require
specialized costs that are not covered by the classics. For instance,
for weighted model counting there are costs associated with the
``CNF-tree'' of the
formula~\cite{DBLP:conf/sum/KenigG15,DBLP:conf/wg/GottlobGMSS05}.  In
the work of Kalinsky et al.~\cite{DBLP:conf/edbt/KalinskyEK17} on
database join optimization, the execution cost is dominated by the
effectiveness of the \e{adhesions} (intersection of neighboring bags)
for caching, particularly the associated \e{skew}. They show real-life
scenarios where isomorphic tree decompositions (of minimum width)
feature orders-of-magnitude difference in performance.
Mediero~\cite{mediero2017search} aims at minimizing the size of AND/OR
trees by seeking decompositions of a low \e{height}.  Abseher et
al.~\cite{DBLP:journals/jair/AbseherMW17} designed a machine-learning
framework to learn the cost function of a tree decomposition in
various problems, using various features of the tree decomposition.

Motivated by the above need, Carmeli et
al.~\cite{DBLP:conf/pods/CarmeliKK17} embarked on the challenge of
\e{enumerating} tree decompositions; that is, generating tree
decompositions one by one so that an application can stop the
enumeration at any time and select the decomposition that best suits
its needs.  As they point out, it is essential to avoid of
\e{redundancy}.  For example, if a graph is already a tree, then there
is no need to further group its vertices. Hence, following Carmeli et
al.~\cite{DBLP:conf/pods/CarmeliKK17}, we consider the task of
enumerating the \e{proper} tree decompositions, which are intuitively
the ones that cannot be improved by splitting a bag or removing it
altogether. They showed that these tree decompositions are precisely
the \e{clique trees} of the \e{minimal triangulations}.  A
\e{triangulation} of a graph $G$ is a chordal graph $H$ obtained from
$G$ by adding edges, called \e{fill edges}.  A triangulation $H$ is
\e{minimal} if no triangulation $H'$ has a strict subset of the fill
edges. Carmeli et al.~\cite{DBLP:conf/pods/CarmeliKK17} proved that
for enumerating tree decompositions, it suffices to enumerate the
minimal triangulations.  See Mediero~\cite{mediero2017search} for an
application of Carmeli et al.~\cite{DBLP:conf/pods/CarmeliKK17}.

While algorithms for generating pools of tree decompositions have been
proposed in the past for small graphs (representing database
queries)~\cite{DBLP:conf/sigmod/TuR15}, Carmeli et
al.~\cite{DBLP:conf/pods/CarmeliKK17} have presented the first
algorithm that has both completeness and efficiency guarantees; that
is, it can generate \e{all} minimal triangulations (and by implication
all proper tree decompositions), and it does so in \e{incremental
  polynomial time}, which means that the time taken for producing the
$N$th result is polynomial in $N$ and in the size of the
input~\cite{DBLP:journals/ipl/JohnsonP88}.  Nevertheless, there can be
exponentially many minimal triangulations, and an effective
enumeration needs to produce earlier the triangulations that are
likely to be low cost. In other words, we would like the algorithm to
enumerate the minimal triangulations by increasing relevant cost such
as width (of some version) or fill-in. In turn, the application will
evaluate the complex cost function on each generated triangulation,
and at a point of choice it will stop the enumeration and pick the
best result found.  Carmeli et al.~\cite{DBLP:conf/pods/CarmeliKK17}
use heuristics to affect the enumeration order, but provide no
guarantees. Without making assumptions it is impossible to guarantee
efficient ranked enumeration for costs, as it is already NP-hard to
compute the first (best) triangulation.

Yet, for some classes of graphs, there is a polynomial-time algorithm
for computing a tree decomposition of a minimum weight/fill-in. These
include the (weakly) chordal graphs, interval graphs, circular-arc
graphs, and cographs. These examples have the property we refer to
\e{poly-MS}: having a polynomial 
number of \e{minimal
  separators}~\cite{DBLP:journals/siamcomp/FominTV15}.  A minimal
separator of a graph is a vertex set $S$ such that some vertices $u$
and $v$ are separated by $S$, but no proper subset of $S$ separates
$u$ from $v$. Various problems have been studied in the context of the
poly-MS assumption~\cite{DBLP:journals/siamcomp/FominTV15}, including
graph isomorphism~\cite{DBLP:conf/swat/OtachiS14}.  Assuming poly-MS,
a tree decomposition of a minimum weight or fill-in can be computed in
polynomial time~\cite{mintriangalg,DBLP:journals/tcs/BouchitteT02}.

The decomposition algorithm of Bouchitt{\'e} and
Todinca~\cite{mintriangalg,DBLP:journals/tcs/BouchitteT02} consists of
two main steps. First, they construct the set of minimal separators of
the input graph, for example using the algorithm of Berry et
al.~\cite{conf/wg/BerryBC99}, and from these compute the set of  all
\e{potential maximal cliques} (which are essentially the bags of the
proper tree decompositions)~\cite{DBLP:journals/tcs/BouchitteT02}.
Second, they use the potential maximal cliques in order to find an optimal
triangulation. In fact, their algorithm has two variants---one for
minimal width (\e{tree-width}) and one for minimal fill-in. The second step has
been later generalized to allow for positive weights on bags (in the case
of width) and edges (in the case of fill) by Furuse and
Yamazaki~\cite{DBLP:journals/tcs/FuruseY14}, again presenting two
corresponding variants of their algorithm.

Our first contribution is a generalization of the concepts of width
and fill-in to general cost functions over tree decompositions. These
cost functions satisfy two properties. First, they assign the same
cost to tree decompositions with the same bags; hence, these are
essentially costs over the set of bags. Second, and more importantly,
they are \e{monotonic} in the following (informal) sense. Suppose that
we cut a tree decomposition $\T$ along an edge, and replace one of the
sides with an alternative subtree (which is a tree decomposition of a
subgraph of the original graph), resulting in a tree decomposition
$\T'$; if the altenative subtree does not cost more than the one it
replaced, then the cost of $\T'$ is no greater than that of $\T$. We
call such a cost function \e{split monotone}, and refer the reader to
Section~\ref{sec:monotonic} for the precise definition.
Split-monotone cost functions generalize existing costs such as
fill-in, width and generalized/fractional hypertree width, as well as
the weighted width and fill-in of Furuse and
Yamazaki~\cite{DBLP:journals/tcs/FuruseY14}. Moreover, we can come up
with various motivated split-monotone costs that are not among the
classic ones, such as the sum over the (exponents of the) bag
cardinalities and linear combinations of width and fill-in.  We
present a generalization of the algorithm of Bouchitt{\'e} and
Todinca~\cite{mintriangalg} to general split-monotone cost
functions. As we explain later, the importance of supporting general
cost functions is not just for the sake of a richer costs; even if we
are interested just in width or fill-in, we need the flexibility of
the cost function in order to incorporate \e{constraints} that we
later use to devise our algorithm for ranked enumeration.

Our main theoretical contributions are algorithms that enumerate
minimal triangulations by increasing cost, for any split-monotone cost
function that is polynomial-time computable (e.g., the aforementioned
ones).  The first algorithm enumerates \e{all} minimal triangulations,
and does so with polynomial delay if the input is from a poly-MS class
of graphs. The second enumerates all minimal triangulations of a
\e{bounded width}, and it does so with polynomial delay if the bound
on the width is a fixed constant. 
\e{Polynomial delay}~\cite{DBLP:journals/ipl/JohnsonP88} means that
the time between every two consecutive answers is polynomial in the
size of the input (graph), a guarantee that is stronger than
incremental polynomial time.  Due to the previously discussed
connection between proper tree decompositions and minimal
triangulations, we get algorithms with the same guarantees for the
enumeration of proper tree decompositions.  Observe that these
algorithms imply polynomial-time procedures for computing top-$k$
minimal triangulations and/or proper tree decompositions.  To the best
of our knowledge, these are the first enumeration algorithms for
minimal triangulations (and proper tree decompositions) with
completeness, efficiency, \e{and order} guarantees.


Our ranked enumeration algorithm adapts the generic procedure of
Lawler-Murty~\cite{LAWLER,MURTY}.  For this
deployment, we use a result by Parra and
Scheffler~\cite{DBLP:journals/dam/ParraS97} who show that a minimal
triangulation is fully identified by its set of minimal separators.
To adopt Lawler-Murty, we rephrase the task as that of enumerating the
relevant sets of minimal separators. In turn, using the technique of
Lawler-Murty we reduce this enumeration to the task of finding an
optimal minimal triangulation, under the cost function, constrained on
\e{including} a given set of minimal separators and \e{excluding} a
given set of other minimal separators.
Towards that, we show that these constraints can be compiled into any
split-monotone cost function so that the resulting cost remains split
monotone. Furthermore, if the original cost function can be computed
in polynomial time, then so can the new cost function with the
constraints compiled in.

\eat{
This
procedure can be described abstractly as follows. There is a set of
items, and the goal of the procedure is to enumerate itemsets by
increasing cost. To do so, the procedure assumes that one can compute
in polynomial time a lowest-cost itemset subject to \e{constraints},
and there are two types of constraints: \e{inclusion constraint}---a
specific item needs to be present in the itemset, and an \e{exclusion
  constraint}---the item needs to be absent. So, for our deployment,
we need to define what the items and itemsets are, and we need to
solve the corresponding constrained optimization problem.

Here, we use a result by Parra and
Scheffler~\cite{DBLP:journals/dam/ParraS97} who show that a minimal
triangulation is fully identified by its set of minimal separators.
Moreover, due to Rose~\cite{rose1970triangulated} it is known that a
minimal triangulation has fewer minimal separators than vertices. Hence,
to adopt Lawler-Murty we define items as vertex sets, and itemsets as
the collections of minimal separators of the minimal
triangulations. To efficiently solve the constrained optimization
problem, we show that inclusion and exclusion constraints on minimal
separators can be complied into any split-monotone cost function so
that the resulting cost remains split monotone. Furthermore, if the
original cost function can be computed in polynomial time, then so can
the new cost function with the constraints compiled in.
}

Finally, we describe an implementation of our algorithm and an
experimental study. We conduct experiments over the datasets of
Carmeli et al.~\cite{DBLP:conf/pods/CarmeliKK17} that consist of three
types of graphs: probabilistic graphical models (from the 2011
\e{Probabilistic Inference Challenge}), database queries (TPC-H), and
random (Erd\H{o}s-R\'enyi) graphs. We also conduct experiments on
graphs from the PACE 2016 competition on tree-width computation.  We
compare our algorithm to the enumeration of Carmeli et
al.~\cite{DBLP:conf/pods/CarmeliKK17} on both the execution time and
the quality (width/fill) of the generated triangulations. In addition,
we explore the validity of the poly-MS assumption on our datasets;
that is, we provide statistics on the number of minimal separators,
and explore the portion of the instances where this number is
``manageable.''

The remainder of the paper is organized as follows. We first present
preliminary definitions and terminology in
Section~\ref{sec:preliminaries}. In Section~\ref{sec:monotonic}, we
describe the notion of split monotonicity. We give our main
theoretical results in Section~\ref{sec:theory}, and present the
algoritms that realize the results
Sections~\ref{sec:mintriang}--\ref{sec:enum}. Specifically,
Section~\ref{sec:mintriang} presents our algorithm for computing a
minimum-cost minimal triangulation, and Section~\ref{sec:enum}
discusses the adaptation of Lawler-Murty to our enumeration. Finally,
we describe our implementation and experimental study in
Section~\ref{sec:exper}, and conclude in
Section~\ref{sec:conclusions}. Due to space limitation, some proofs
are in the Appendix.

\section{Preliminaries}\label{sec:preliminaries}
We begin by introducing the basic notation, terminology and formal
concepts that we use throughout the paper. For convenience,
Table~\ref{table:notation} contains important notation that we present
here and later.

\partitle{Graphs and Cliques}
All the graphs in this paper are undirected. We denote by $\nodes(G)$
and $\edges(G)$ the set of vertices and edges, respectively, of a
graph $G$. An edge in $\edges(G)$ is a pair $\set{u,v}$ of distinct
vertices in $\nodes(G)$. The \e{union} of two graphs $G_1$ ad $G_2$,
denoted $G_1 \cup G_2$, is the graph $G$ with
$\nodes(G)=\nodes(G_1)\cup\nodes(G_2)$ and
$\edges(G)=\edges(G_1)\cup\edges(G_2)$.

A set $C$ of vertices of a graph $G$ is a \e{clique} (\e{of $G$}) if
every two vertices in $C$ are connected by an edge of $G$. The set $C$
is a \e{maximal clique} (\e{of $G$}) if $C$ is not strictly contained
in any other clique of $G$. We note by $\maxcliques(G)$ the set of all 
maximal cliques of $G$. 

Let $G$ be a graph, and $U$ a set of vertices of $G$. We denote by
$\K_U$ is the \e{complete} graph over a vertex set $U$; that is,
$\K_U$ is the graph with $\nodes(K_U)=U$ and $\edges(\K_U) =
\set{\set{u,v}\subseteq U\mid u\neq v}$ (hence, $U$ itself is a clique
of $\K_U$). By \e{saturating $U$ (in $G$)} we refer to the operation
connecting every non-adjacent vertices in $U$ by a new edge, thereby
making $U$ a clique of $G$. In other words, saturating $U$ refers to
the operation of replacing $G$ with $G\cup\K_U$.


A \e{subgraph} of a graph $G$ is a graph $G'$ with
$\nodes(G')\subseteq\nodes(G)$ and $\edges(G')\subseteq\edges(G)$.
Let $U \subseteq \nodes(G)$ be a set of vertices of $G$. We denote by
$\induced{G}{U}$ the subgraph of $G$ that is \e{induced} by $U$; that
is, $\induced{G}{U}$ is the graph $G'$ with $\nodes(G')=U$ and
$\edges(G')=\set{e\in\edges(G)\mid e\subseteq U}$.  Let $G$ be a
graph, and $U$ a set of vertices of $G$.  We denote by $G \setminus U$
the graph obtained from $G$ by removing all vertices in $U$ (along
with their incident edges); that is, $G\setminus U$ is the graph
$\induced{G}{\nodes(G)\setminus U}$.

\partitle{Tree decompositions}
A \e{tree decomposition} $\T$ of a graph $G$ is a pair $(T,\beta)$,
where $T$ is a tree and $\beta:\nodes(T)\rightarrow2^{\nodes(G)}$ is a
function that maps every node of $T$ to a set of vertices of $G$, so
that all of the following hold.
\begin{citemize}
\item Vertices are covered: for every vertex $u$ of $G$ there is a
  vertex $v$ of $T$ such that $u\in\beta(v)$.
\item Edges are covered: for every edge $e$ of $G$ there is a vertex
  $v$ of $T$ such that $e\subseteq\beta(v)$.
\item The \e{junction-tree} property: for all vertices $u$ and $v$ of
  $T$, the intersection $\beta(u)\cap\beta(v)$ is contained in every
  vertex along the path between $u$ and $v$.
\end{citemize}

Let $G$ be a graph, and let $\T=(T,\beta)$ be a tree decomposition of
$G$. A set $\beta(v)$, for $v\in\nodes(T)$, is called a \e{bag} of
$\T$. We denote by $\bags(\T)$ the set $\set{\beta(v)\mid
  v\in\nodes(\T)}$. 

Let $\T_1=(T_1,\beta_1)$ and $\T_2=(T_2,\beta_2)$ be two tree
decompositions of a graph $G$. We say that $\T_2$ \e{bag-contains}
$\T_1$ if there is an injection
$\varphi:\nodes(T_1)\rightarrow\nodes(T_2)$ such that
$\beta_1(v)=\beta_2(\varphi(v))$ for all $v\in\nodes(T_1)$.  We say
that $\T_1$ and $\T_2$ are \e{bag equivalent} if $\T_2$ bag-contains
$\T_1$ and vice versa. 
\noamdone{Switch the rest with a shorter definition of a proper tree
  decomposition?}  We say that $\T_1$ \e{strictly subsumes} $\T_2$ if
$\T_1$ is obtained from $\T_2$ by splitting a bag or removing it
altogether. More formally, $\T_1$ strictly subsumes $\T_2$ if there is
a mapping $\varphi:\nodes(T_1)\rightarrow\nodes(T_2)$ such that
$\beta_1(x)\subseteq\beta_2(\varphi(x))$ for all $x\in\nodes(T_1)$,
and for at least one $y\in\nodes(T_2)$ it is the case that
$\beta_1(x)\subsetneq \beta_2(y)$ whenever $\varphi(x)=y$ (hence, either no
vertex is mapped to $y$ or a vertex that is mapped to $y$ is a strict
subset of $y$). A tree decomposition is \e{proper} if it is not
strictly subsumed by any other tree
decomposition~\cite{DBLP:conf/pods/CarmeliKK17}.\footnote{This is a
  simplified, yet equivalent definition to that of Carmeli et
  al.~\cite{DBLP:conf/pods/CarmeliKK17}.}

\begin{example}\label{example:tds}
  Figure~\ref{fig:pre-tds} depicts five tree decompositions of the
  graph $G$ of Figure~\ref{fig:pre-example}. Each rectangle
  corresponds to a vertex $x$ of the tree, and the bag $\beta(x)$ is
  depicted inside the rectangle. For example, if we denote
  $\T_1=(T_1,\beta_1)$, then $T_1$ is a path of three vertices
  (corresponding to the three rectangles), and for the top vertex $x$
  we have $\beta_1(x)=\set{u,w_1,w_2,w_3}$.  Note that $\T_2$ and
  $\T_2''$ are bag equivalent, since they have the exact same bags
  (though connected differently). The tree decomposition $\T_1$
  strictly subsumes $\T_1'$, as the latter is obtained from the former
  by adding $w_1$ to the bottom bag. Moreover, $\T_2$ strictly
  subsumes $\T_2'$, as the former is obtained from the latter by
  splitting the bottom bag into two---the middle and bottom vertices
  fo $\T_2$. Thus, $\T_1'$ and $\T_2'$ are not proper. We will later
  show that $\T_1$ and $\T_2$ (and, hence, $\T''_2$) are proper.\qed
\end{example}

\begin{figure}
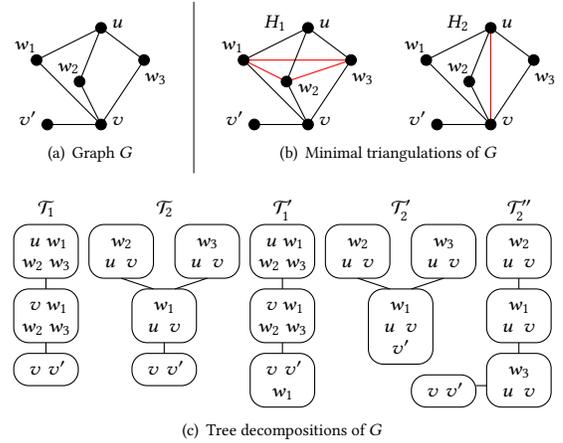

\scalebox{0.9}{
\centering
\subfigure[\label{fig:pre-example}Graph $G$]{\input{pre-example.pspdftex}}\quad\quad\vline\quad\quad
\subfigure[\label{fig:pre-mintriangs}Minimal triangulations of $G$]{\input{pre-triangs-small.pspdftex}}}
\vskip0.5em
\scalebox{0.9}{\subfigure[\label{fig:pre-tds}Tree decompositions of $G$]{\input{pre-tds.pspdftex}}}
\caption{A graph $G$, tree decompositions of $G$ and minimal
triangulations of $G$.}
\end{figure}

\partitle{Minimal triangulations}
Let $G$ be a graph. A \e{cycle} in $G$ is a path that starts and ends
with the same vertex. A \e{chord} of a cycle $C$ is an edge
$e\in\edges(G)$ that connects two vertices that are non-adjacent in
$C$. We say that $G$ is \e{chordal} if every cycle of length greater
than three has a chord. Whether a given graph is chordal can be
decided in linear time~\cite{Tarjan:1984:SLA:1169.1179}.

A \e{triangulation} of a graph $G$ is a chordal graph $H$ that is
obtained from $G$ by adding edges. The \e{fill set} of a triangulation
$H$ of $G$ is the set of edges added to $H$, that is,
$\edges(H)\setminus\edges(G)$. A \e{minimal triangulation} of $G$ is a
triangulation $H$ of $G$ such that the fill set of $H$ is not strictly
contained in the fill set of any other triangulation; that is, there
is no chordal graph $G'$ with $\nodes(G)=\nodes(G')$ and
$\edges(G)\subseteq\edges(G')\subsetneq\edges(H)$.  In particular, if
$G$ is already chordal then $G$ is the only minimal triangulation of
itself. 

{\begin{table} \renewcommand{\arraystretch}{1.2}
\caption{\label{table:notation}Symbol table}
\small
\begin{tabular}{c|l}\hline
\textbf{Notation} & \multicolumn{1}{c}{\textbf{Meaning}}\\\hline
$\nodes(G)$ / $\edges(G)$ & Vertex/edge set of $G$\\
$\maxcliques(G)$ & Maximal cliques of $G$\\
$\K_U$ & Full graph (clique) over node set $U$\\
\hline
$\minseps(G)$ & Minimal separators of $G$\\
$(S,C)$ & Block of $G$\\
$\R(S,C)$ & Realization of $(S,C)$, that is, $G[S\cup C]\cup \K_S$\\\hline
$\pmcs(G)$ & Potential maximal cliques of $G$\\
$\minseps_G(\Omega)$ & Minimal separators associated to $\Omega$ of $G$\\
$\fblocks_G(\Omega)$ & Blocks associated to $\Omega$ in $G$\\
\hline
\end{tabular}
\end{table}}

\partitle{Clique trees}
Let $G$ be a graph. A \e{clique tree} of $G$ is a tree decomposition
$\T=(T,\beta)$ of $G$ such that $\beta$ is bijection between
$\nodes(T)$ and $\maxcliques(G)$. In other words, $\T$ is a clique
tree of $G$ if $\bags(\T)=\maxcliques(G)$ and no two bags are the
same. The following is known, and recorded for later use.

\begin{citedtheorem}{\cite{ChordalIntroduction,rose1970triangulated,DBLP:conf/pods/CarmeliKK17}}\label{thm:clique-tree-essentials}
  Let $G$ be a graph.
\begin{cenumerate}
\item $G$ is chordal if and only if $G$ has a clique
  tree~\e{\cite{ChordalIntroduction}}.
\item If $G$ is chordal, then
  $|\maxcliques(G)|<|\nodes(G)|$~\e{\cite{rose1970triangulated}}.
\item A tree decomposition $\T$ of $G$ is proper if and only if it is
  a clique tree of a minimal triangulation of
  $G$~\e{\cite{DBLP:conf/pods/CarmeliKK17}}.
\end{cenumerate}
\end{citedtheorem}

\begin{example}\label{example:triangs}
  Continuing Example~\ref{example:tds}, observe that the graph $G$ of
  Figure~\ref{fig:pre-example} is not chordal. As one evidence, it has
  the chordless cycle
  $u\mbox{---}w_1\mbox{---}v\mbox{---}w_2\mbox{---}u$.
  Figure~\ref{fig:pre-mintriangs} depicts two minimal triangulations,
  $H_1$ and $H_2$, of the graph $G$ of Figure~\ref{fig:pre-example}.
  The reader can verify that $\T_1$ and $\T_2$ of
  Figure~\ref{fig:pre-tds} are clique trees of $H_1$ and $H_2$,
  respectively. In particular, we conclude that $\T_1$ and $\T_2$ are
  proper tree decompositions.  \qed\end{example}

\partitle{Minimal separators}
\noamdone{Should this be here or earlier in the preliminaries?}  Let
$u$ and $v$ be vertices of a graph $G$. A \e{$(u,v)$-separator
  (w.r.t.~$G$)} is a set $S\subseteq\nodes(G)$ such that $u$ and $v$
belong to different connected components in $G\setminus S$; that is,
$G\setminus S$ does not contain any path between $u$ and $v$ (or
equivalently, every path between $u$ and $v$ visits one or more
vertices of $S$). We say that $S$ is a \e{minimal $(u,v)$-separator}
if no proper subset of $S$ is a $(u,v$)-separator. We say that $S$ is
a \e{minimal separator} of $G$ if there are vertices $u$ and $v$ such
that $S$ is a minimal $(u,v)$-separator. We denote by $\minseps(G)$
the set of all minimal separators of $G$. A graph may have exponentially 
many minimal separators.

Let $G$ be a graph, and let $S$ and $T$ be two minimal separators of
$G$. We say that $S$ \e{crosses} $T$ if there are vertices $u$ and $v$
in $T$ such that $S$ is a $(u,v)$-separator.  Crossing is known to be
a symmetric relation: if $S$ crosses $T$ then $T$ crosses
$S$~\cite{DBLP:journals/dam/ParraS97,DBLP:journals/tcs/KloksKS97}. Hence,
if $S$ crosses $T$ then we may also say that $S$ and $T$ are
\e{crossing}. When $S$ and $T$ are non-crossing, then we also say that
$S$ and $T$ are \e{parallel}. 

\begin{example}\label{example:min-seps}
  We continue with our running example. The top-left graph of
  Figure~\ref{fig:seps-real} depicts three minimal separators $S_1$,
  $S_2$ and $S_3$ of the graph $G$ of
  Figure~\ref{fig:pre-example}. For instance $S_1=\set{w_1,w_2,w_3}$
  is a minimal $(u,v)$-separator, $S_2=\set{u,v}$ is a minimal
  $(w_1,w_2)$-separator, and $S_3=\set{v}$ is a minimal
  $(u,v')$-separator. Note that $S_2$ is a $(w_1,v')$-separator but not
  a minimal $(w_1,v')$-separator, since a strict subset of $S_2$, namely
  $S_3$, is a $(w_1,v')$-separator. Also note that $S_1$ and $S_2$ are
  crossing, since $S_1$ is a $(u,v)$-separator (and also $S_2$ is a
  $(w_1,w_2)$-separator).  It can be verified that $S_1$, $S_2$ and
  $S_3$ are the only minimal separators of $G$. Hence
  $\minseps(G)=\set{S_1,S_2,S_3}$. This example illustrates that,
  albeit being ``minimal,'' a minimal separator can be a strict subset
  of another; for instance $S_3\subsetneq S_2$.
\qed\end{example}

A set of \e{pairwise-parallel} minimal
separators is a set $M$ of minimal separators such that every two
distinct members of $M$ are parallel. Moreover, such $M$ is said to be
\e{maximal} if every minimal separator not in $M$ is crossing at least
one member of $M$. Parra and
Scheffler~\cite{DBLP:journals/dam/ParraS97} established the following
connection between minimal triangulations and maximal sets of
pairwise-parallel minimal separators.

\noamdone{I added this Theorem here. It is the only one I chose to keep
from section 5, maybe you could shorten the wording.}

\begin{citedtheorem}{Parra and
    Scheffler~\cite{DBLP:journals/dam/ParraS97}}\label{thm:ParraS97}
  \!\!Let $G$ be a graph.
	\begin{cenumerate}
        \item \label{item:ParraS1st} Let $M$ be a maximal set of
          pairwise-parallel minimal separators of $G$, and let $H$ be
          obtained from $G$ by saturating each member of $M$.  Then
          $H$ is a minimal triangulation of $G$ having
          $\minseps(H)=M$.
        \item \label{item:ParraS2nd} Conversely, if $H$ is a minimal
          triangulation of $G$, then $M=\minseps(H)$ is a maximal set
          of pairwise-parallel minimal separators in $G$, and $H$ is
          obtained from $G$ by saturating each member of $M$. 
	\end{cenumerate}
\end{citedtheorem}

\partitle{Ranked enumeration}
An \e{enumeration problem} $\P$ is a collection of pairs $(x,Y)$ where
$x$ is an \e{input} and $Y$ is a finite set of \e{answers} for $x$,
denoted by $\P(x)$. A \e{solver} for an enumeration problem $\P$ is an
algorithm that, when given an input $x$, produces (or \e{prints}) a
sequence of answers such that every answer in $\P(x)$ is printed
precisely once.  A solver for an enumeration problem is also called an
\e{enumeration algorithm}. We recall several standard yardsticks of
efficiency for enumeration
algorithms~\cite{DBLP:journals/ipl/JohnsonP88}. Let $\P$ be an
enumeration problem, and let $A$ be solver for $\P$. We say that $A$
runs in
\e{(a)}
\e{polynomial total time} if the total execution time of $A$ is
  polynomial in $(|x|+|\P(x)|)$;
\e{(b)}
\e{polynomial delay} if the time between printing every two
  consecutive answers is polynomial in $|x|$;
\e{and (c)}
\e{incremental polynomial time} if, after printing a sequence
  $Y$ of answers, the time to print the next answer is polynomial in
  $(|x|+|Y|)$ where $Y$ is the size of the representation of $Y$.
Observe that a solver that enumerates with polynomial delay also
enumerates with incremental polynomial time, which, in turn, implies
polynomial total time. 

Let $\P$ be an enumeration problem. A \e{cost function} for $\P$ is a
function $c$ that associates a numerical cost $c(x,y)$ to each input
$x$ and answer $y$ for $x$. A solver $A$ for $\P$ is said to
\e{enumerate by increasing $c$}, where $c$ is cost function for $\P$,
if for every two answers $a_1$ and $a_2$ produced by $A$, if $a_1$ is
produced before $a_2$ then $c(a_1)\leq c(a_2)$.

\section{Monotone Cost Functions}\label{sec:monotonic}

By a \e{cost function over tree decompositions} we refer to a function
$\kappa$ that maps a graph $G$ and a tree decomposition $\T$ for $G$ to
a numerical (positive, negative or zero) value $\kappa(G,\T)$. In this
section we define a class of such cost functions that includes many of
the common costs such as width and fill. This class is defined by
means of monotonicity, as we formally defined next.

Let $G$ be a graph, and let $\T=(T,\beta)$ be a tree decomposition of
$G$. Every edge $e=\set{v_1,v_2}$ of $T$ connects two unique subtrees
of $T$---one connected to $v_1$ and one connected to $v_2$. Let $e$ be
an edge of $T$, let $T_1$ and $T_2$ be the two subtrees connected by
$e$, and let $\beta_1$ and $\beta_2$ be the restrictions of $\beta$ to
$\nodes(T_1)$ and $\nodes(T_2)$, respectively. Let
$\T_1=(T_1,\beta_1)$ and $\T_2=(T_2,\beta_2)$, and let $G_1$ and $G_2$
be the subgraphs of $G$ induced by the nodes in the bags of $\T_1$
and $\T_2$, respectively. Then we say that $\T$ \e{splits} (\e{by
  $e$}) as $\angs{G_1,\T_1,G_2,\T_2}$. The following proposition is
straightforward.

\begin{proposition}\label{prop:split}
  Let $G$ be a graph and $\T=(T,\beta)$ a tree decomposition of $G$.
  If $\T$ splits as $\angs{G_1,\T_1,G_2,\T_2}$, then $\T_1$ and $\T_2$
  are tree decompositions of $G_1$ and $G_2$, respectively.
\end{proposition}

Proposition~\ref{prop:split} implies that if $\kappa$ is a cost
function and $\T$ splits as $\angs{G_1,\T_1,G_2,\T_2}$, then both
$\kappa(G_1,\T_1)$ and $\kappa(G_2,\T_2)$ are defined. We can now
define properties of cost functions.

\begin{definition}\label{def:cost-props}
  Let $\kappa$ be a cost function over tree decompositions. We say
  that $\kappa$ is:
  \begin{cenumerate}
  \item \e{invariant under bag equivalence}, if for all graphs $G$ and
    tree decompositions $\T_1$ and $\T_2$ of $G$, if $\T_1$ and $\T_2$
    are bag equivalent then $\kappa(\T_1)=\kappa(\T_2)$.
  \item \e{split monotone} if for all graphs $G$ and tree
    decompositions $\T$ and $\T'$ of $G$, if $\T$ and $\T'$ split as
    $\angs{G_1,\T_1,G_2,\T_2}$ and $\angs{G_1,\T_1',G_2,\T_2'}$,
    respectively, and $\kappa(G_i,\T_i)\leq\kappa(G_i,\T_i')$ for
    $i=1,2$, then $\kappa(G,\T)\leq\kappa(G,\T')$.
  \end{cenumerate}
\end{definition}

If $\kappa$ is invariant under bag equivalence, then it is essentially
a scoring function over the collection of bags, and in that case we
say that $\kappa$ is a \e{bag cost}.

For illustration, the following most popular cost functions
$\kappa(G,\T)$ are both split-monotone bag costs.
\begin{itemize}
\item $\width(G,\T)$: the maximal bag cardinality minus one.
\item $\fillin(G,\T)$: the number of edges required to be added for
  saturating all bags.
\end{itemize}
Other such cost functions are the generalizations of width and fill-in
introduced by Furuse and Yamazaki~\cite{DBLP:journals/tcs/FuruseY14},
where it is assumed that each bag $b$ has a cost $c(b)$, and each edge
$e$ has a cost $c(e)$. Then, they define $\width_c(G,\T)$ to be
the maximal score of a bag, and $\fillin_c(G,\T)$ to be the sum
of costs of the edges required to saturate all bags. As a special
case, if the graph $G$ is the primal graph of a hypergraph, then
$c(b)$ can be the minimal number of hyperedges needed to cover $b$, or
the minimal weight of a \e{fractional} edge cover of $b$, thereby
establishing the popular cost functions of \e{hypertree
  width}~\cite{DBLP:conf/pods/GottlobLS99} and \e{fractional}
hypertree width~\cite{DBLP:journals/talg/GroheM14}.

Finally, another intuitive split-monotone bag costs is
\[\kappa(G,\T)=|\edges(G)|^{\width(G,\T)}+\fillin(G,\T)\]
that effectively establishes the lexicographic ordering of the width
followed by the fill-in of $G$. 

We can then use a bag cost $\kappa$ as a cost function over
triangulations $H$, by defining the cost as $\kappa(\T)$ where $\T$ is
\e{any} clique tree of $H$. Since $\kappa$ is invariant under bag
equivalence (being a bag cost), then the choice of $\T$ does not
matter. By a slight abuse of notation, we use $\kappa(G,H)$ to denote
the resulting cost over triangulations $H$ of $G$.

\section{Main Theoretical Results}\label{sec:theory}

In this section we present our main theoretical results. These results
are upper bounds (existence of algorithms) on problems of ranked
enumeration of tree decompositions and minimal triangulations. In the
next two sections we will describe the algorithms.

\partitle{The poly-MS assumption}
Recall that a graph may have an exponential number of minimal
separators. Our main result holds for the case where this number is
reasonable, a case that was deeply investigated in past
research~\cite{DBLP:journals/tcs/FuruseY14,DBLP:journals/tcs/BouchitteT02,DBLP:conf/wg/MontealegreT16,
  DBLP:conf/wg/LiedloffMT15}. Formally, we consider classes $\G$ of
graph such that for some polynomial $p$ it is the case that
$|\minseps(G)|\leq p(|\nodes(G)|)$ for all $G\in\G$.  We then say
shortly that $\G$ is a \e{poly-MS} class of graphs (where ``MS''
stands for \e{Minimal Separators}).  Later in this paper we
empirically study the applicability of this assumption on real and
synthetic datasets.
\noamdone{Should this paragraph move to the preliminaries?}

\partitle{State of the art}
Before presenting our results, we recall some relevant results from
the literature.  Carmeli et al.~\cite{DBLP:conf/pods/CarmeliKK17}
showed that, without making any assumption, one can enumerate in
incremental polynomial time the set of all proper tree decompositions
and the set of all minimal triangulations. Note, however, that no
guarantee is made on the order of enumeration.

\begin{citedtheorem}{\cite{DBLP:conf/pods/CarmeliKK17}\label{thm:enum-incp}}
Given a graph $G$, one can enumerate in incremental polynomial time all
proper tree decompositions, and all minimal triangulations.
\end{citedtheorem}

Parra and Scheffler~\cite{DBLP:journals/dam/ParraS97} showed that
minimal triangulations are in one-to-one correspondence with the
maximal independent sets of the graph that has the minimal separators
as vertices, and an edge between every two crossing
separators. Combining that with results on the enumeration of maximal
independent
sets~\cite{DBLP:journals/ipl/JohnsonP88,DBLP:journals/jcss/CohenKS08},
we get that that for poly-MS classes of graphs, the
minimal triangulations can be enumerated with polynomial delay (again with
no guarantees on the order). Moreover, Carmeli et 
al.~\cite{DBLP:conf/pods/CarmeliKK17} show that such enumeration automatically
translates into an algorithm for enumerating the proper tree decompositions with
polynomial delay. Hence, we get the following.

\begin{citedtheorem}{see~\cite{DBLP:conf/pods/CarmeliKK17}}\label{thm:p-delay}
  \!\!If $\G$ is a poly-MS class of graphs, then one can enumerate with
  polynomial delay all proper tree decompositions, and all minimal
  triangulations.
\end{citedtheorem}

Bouchitt{\'{e}} and Todinca showed that on poly-MS classes of graphs,
a tree decomposition (or triangulation) of a minimal width or fill-in
can be found in polynomial time.

\begin{citedtheorem}{\cite{DBLP:journals/tcs/BouchitteT02}} \label{thm:fill-width-min}
  Let $\G$ be a poly-MS class of graphs. One can find in polynomial
  time a minimal-cost tree decomposition (or triangulation) when the
  cost is either the width or the fill in.
\end{citedtheorem}
Furuse and Yamazaki~\cite{DBLP:journals/tcs/FuruseY14} generalized the
above result to the cost functions $\width_c$ and $\fillin_c$ as
defined in Section~\ref{sec:monotonic}.

\partitle{Main Results} We now turn to our results. The main result
generalizes Theorems~\ref{thm:p-delay} and~\ref{thm:fill-width-min} in
two directions. First, the enumeration is ranked. Second, the cost
function is not just width of fill-in, but in fact every bag cost that
is split monotone and computable in polynomial time.

\begin{theorem}\label{thm:enum-ranked}
  Let $\G$ be a poly-MS class of graphs, and let $\kappa$ be a bag cost
  that is split monotone and computable in polynomial time. On graphs
  of $\G$ one can enumerate with polynomial delay
  all: \begin{cenumerate}
  \item proper tree decompositions $\T$
  by increasing $\kappa(G,\T)$;
\item minimal triangulations $H$ by
  increasing $\kappa(G,H)$.
\end{cenumerate} \end{theorem}

Finally, the next result applies to general graphs, and assumes that
we are interested only in tree decompositions of a bounded width. In
this case, we get a ranked enumeration with polynomial delay
\e{without} assuming an upper bound on the number of minimal separators.

\begin{theorem}\label{thm:enum-ranked-btw}
  Let $b$ be a fixed natural number, and let $\kappa$ be a bag cost
  that is split monotone and computable in
  polynomial time. Given a graph, one can enumerate with polynomial
  delay all:
\begin{cenumerate}
\item proper tree decompositions $\T$ of width at most $b$ by
  increasing $\kappa(G,\T)$;
\item minimal triangulations $H$ of width at most $b$ by increasing
   $\kappa(G,H)$.
\end{cenumerate}
\end{theorem}

As said previously, in the next sections we prove
Theorems~\ref{thm:enum-ranked} and~\ref{thm:enum-ranked-btw} by
presenting the corresponding algorithms and proving their correctness.

\section{Computing an Optimal Minimal
  Triangulation}\label{sec:mintriang}

We now present an algorithm for computing a minimum-cost
minimal triangulation, assuming that the cost function is a
split-monotone bag cost. Our algorithm terminates in polynomial time
if the cost function can be evaluated in polynomial time, and
moreover, the input graphs belong to a poly-MS class of graphs. Our
algorithm generalizes an algorithm by Bouchitt\'e and
Todinca~\cite{mintriangalg} for computing the minimal width (\e{tree-width}) and
minimum fill-in over a poly-MS class of graphs. Furthermore, we will consider 
restriction to triangulations of a bounded width.

\subsection{Definitions and Notation}
We first recall some concepts from the literature that our
triangulation algorithm builds upon. For convenience, some of the
notation is shown in Table~\ref{table:notation}.



\begin{figure}
\vskip-2em
\centering
\scalebox{0.90}{\input{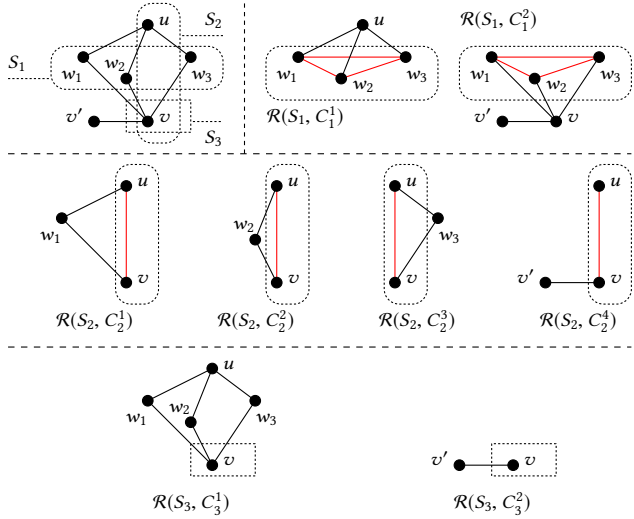}}
\caption{\label{fig:seps-real}Minimal separators and realizations of blocks of
the graph $G$ of Figure~\ref{fig:pre-example}.}
\end{figure}

\partitle{Components and Blocks}
Let $G$ be a graph, and $U$ a set of vertices of $G$. A \e{$U$-component} (\e{of
  $G$}) is a connected component of the graph $G\setminus U$.\eat{We
denote the set of all $U$-components of $G$ by $\comps_G(U)$, or only
$\comps(U)$ if $G$ is clear from the context. } Recall that a
\e{connected component} is a subset $W$ of $\nodes(G)$ such that $G$
contains a path from each vertex of $W$ to every other vertex of $W$, and
to none of the vertices outside $W$.

Let $S$ be a minimal separator of $G$. An $S$-component $C$ is
\e{full} if every vertex in $S$ is connected to one or more vertices
in $C$.\eat{We denote by $\compsstar_G(S)$ the set of full
  $S$-components.} A \e{block} (\e{of $G$}) is a pair $(S,C)$ where
$S$ is a minimal separator and $C$ is an $S$-component. By a slight
abuse of notation, we often identify the block $(S,C)$ with the vertex
set $S\cup C$. A block $(S,C)$ is \e{full} if $C$ is a full
component. The \e{realization} of the block $(S,C)$, denoted
$\R_G(S,C)$, is the induced graph of $(S,C)$ after saturating $S$;
that is:
\[\R_G(S,C)\eqdef G[S\cup C]\cup \K_S\]
When $G$ is clear from the context, we may remove it from the
subscripts and write simply $\R(S,C)$.

\begin{example}\label{example:blocks}
  Recall from Example~\ref{example:min-seps} that for the graph $G$ of
  our running example (Figure~\ref{fig:pre-example}) we have
  $\minseps(G)=\set{S_1,S_2,S_3}$, where the $S_i$ are depicted in
  Figure~\ref{fig:seps-real} on the top left. The rest of
  Figure~\ref{fig:seps-real} shows the different realizations $\R(S_i,C_i^j)$ of
  the blocks of $G$. A rectangle marks the vertices of the minimal separator of
  the block, and the edges that have been added in the saturation are colored red.
  Note that all of the blocks are full, except for $(S_2,C_2^4)$ where no
  vertex of $C_2^4$ is connected to $u$.
\qed\end{example}

\partitle{Potential Maximal Cliques}
Let $G$ be a graph. A vertex set $\Omega \subseteq \nodes(G)$ is a
\textit{Potential Maximal Clique} (\e{PMC} for short) if there is a
minimal triangulation $H$ of $G$ such that $\Omega$ is a maximal
clique of $H$. Due to Theorem~\ref{thm:clique-tree-essentials} we
conclude that a vertex set $\Omega$ is a PMC if and only if it is a
bag of some proper tree decomposition of $G$. We denote by $\pmcs(G)$
the set of PMCs of $G$.

For $\Omega\in\pmcs(G)$, let $C$ be any $\Omega$-component. Let $S$
be the set of all vertices in $\Omega$ that are neighbors of vertices in
$C$. It is known that $S$ is a minimal separator of $G$ and $(S,C)$ is a full
block of $G$ \cite{mintriangalg}, and they are said to be \e{associated to
$\Omega$} (\e{in $G$}). No other minimal separator of $G$ is contained in
$\Omega$. We denote by $\minseps_G(\Omega)$ and $\fblocks_G(\Omega)$ the sets of
minimal separators and blocks, respectively, associated to $\Omega$. When $G$ is clear from the context, we may omit it and write simply $\minseps(\Omega)$ and
$\fblocks(\Omega)$.

\begin{example}\label{example:pmc}
  The PMCs of the graph $G$ of Figure~\ref{fig:pre-example} are the
  vertex sets in all of the bags of the proper tree decompositions in
  Figure~\ref{fig:pre-tds}.  For example, $\pmcs(G)$ contains the sets
  $\set{u_1,w_1,w_2,w_3}$ and $\set{w_1,u,v}$. Let $\Omega$ be
  $\set{w_1,u,v}$. The minimal separators in $\minseps(\Omega)$ are
  $S_2$ and $S_3$ of Figure~\ref{fig:seps-real}, and the blocks in
  $\fblocks(\Omega)$ are $(S_2,C_2^2)$, $(S_2,C_2^3)$, 
  $(S_3,C_3^2)$, all depicted in Figure~\ref{fig:seps-real}.
\qed\end{example}

\def\B{\mathcal{B}}
\def\H{\mathcal{H}}
\begin{algseries}{t}{Computing an optimal minimal triangulation
    for a split-monotone bag cost.\label{alg:triang}}
\begin{insidealg}{MinTriang\langle\kappa\rangle}{$G$}
\STATE Compute $\minseps(G)$ and $\pmcs(G)$
\STATE $\B \asn$ the set of full blocks of $G$ 
\FOR{$(S,C)\in\B$ by increasing cardinality} \label{line:iterate blocks}
\STATE $\Omega(S,C) \asn \underset{\Omega \in \pmcs(S,C)}{\algname{argmin}}
{\kappa(\induced{G}{(S,C)},\HByPMC{\R(S,C)}{\Omega})}$ \label{line:saturated pmc block}
\STATE $\H(S,C) \asn \HByPMC{\R(S,C)}{\Omega(S,C)}$ \label{line:triang
block}
\ENDFOR
\STATE $\Omega(G) \asn \algname{argmin}_{\Omega \in \pmcs(G)}
\kappa(G,\HByPMC{G}{\Omega})$ \label{line:saturated pmc graph}
\STATE $\mathbf{return}(\HByPMC{G}{\Omega(G)})$ \label{line:triang graph}
\end{insidealg}
\end{algseries}

\subsection{Algorithm Description}

To describe the algorithm, we first give some background.  The
Bouchitt\'e-Todinca algorithm~\cite{mintriangalg} is based on the
observation that a minimal triangulation of a graph is composed of
minimal triangulations over realizations of its blocks. This is
formalized in the following theorem:

\begin{citedtheorem}{Bouchitt\'e and Todinca \cite{mintriangalg}}
  \label{theorem:pmc-block-triangulation}
   The following hold for a graph $G$.
\begin{cenumerate}
\item If $H$ is a minimal triangulation, $\Omega \in \maxcliques(H)$,
  and $(S_i,C_i) \in \fblocks_G(\Omega)$, then $H_i=
  \induced{H}{S_i \cup C_i}$ is a minimal triangulation of the
  realization $\R_G(S_i,C_i)$.
\item Conversely, let $\Omega \in \pmcs(G)$ and
  $k=|\fblocks_G(\Omega)|$. For $(S_i, C_i) \in \fblocks_G(\Omega)$
  let $H_i$ be a minimal triangulation of $\R_G(S_i, C_i)$. Then $H =
  \bigcup_{i=1}^k H_i \cup \K_\Omega$ is a minimal triangulation of
  $G$.
\end{cenumerate}
\end{citedtheorem}

Observe that $\Omega \in \maxcliques(H)$ implies that $\Omega$ is a PMC.
Theorem~\ref{theorem:pmc-block-triangulation} provides a
characterization of the minimal triangulations in terms of the minimal
triangulations of the block realizations. Then, how do we proceed to
computing the minimal triangulations of the block realizations? This
is shown in the following result.

\begin{citedtheorem}{Bouchitt\'e and Todinca~\cite{mintriangalg}}
  \label{theorem:realization triangulation}
  Let $G$ be a graph, $S \in \minseps(G)$, and $(S,C)$ a full block of
  $G$. Let $G'=\R_G(S,C)$.  Let $H$ be a graph with
  $\nodes(H)=\nodes(G')$ and $\edges(H)\supseteq\edges(G')$.  The
  following are equivalent.
\begin{cenumerate}
\item $H$ is a minimal triangulation of $G'$.
\item There exists $\Omega \in \pmcs(G)$ such that $S \subset \Omega
  \subseteq \nodes(G')$ and $H=\bigcup_{i=1}^k H_i \cup \K_\Omega$,
  where $|\fblocks_{G'}(\Omega)|=k$ and $H_i$ is a minimal
  triangulation of $\R_{G'}(S_i,C_i)$ for all $(S_i,C_i) \in
  \fblocks_{G'}(\Omega)$.
\end{cenumerate}
\end{citedtheorem}

For Part~2 of Theorem~\ref{theorem:realization triangulation}, it is
important to note that the blocks of $\Omega$ in $G'$ are also full
blocks of $G$~\cite{mintriangalg}.


Now, consider an input $G$ for our algorithm and $\Omega \in
\pmcs(G)$.  We will assume that for each $(S_i,C_i) \in
\fblocks_G(\Omega)$ our algorithm has computed a minimal triangulation
$H_i$ for the realization $\R_G(S_i,C_i)$. We then define the
following.
\begin{gather} 
\HByPMC{G}{\Omega} \eqdef \bigcup_{(S_i,C_i) \in \fblocks_G(\Omega)}H_i\cup
\K_\Omega \label{eq:H by U}
\end{gather}

Our algorithm, $\algname{MinTriang}$, is depicted in
Figure~\ref{alg:triang}. It applies dynamic programming based on
Equation~\eqref{eq:H by U}. The algorithm is parameterized by a
split-monotone bag cost $\kappa$, takes as input a graph $G$, and
computes a minimum-$\kappa$ minimal triangulation of $G$.

The first step is to compute $\minseps(G)$ and $\pmcs(G)$.  Assuming
$G$ belongs to a poly-MS class of graphs, this step can be done
efficiently by combining results from Berry et
al.~\cite{conf/wg/BerryBC99} and Bouchitt\'e and
Todinca~\cite{DBLP:journals/tcs/BouchitteT02}.
Then, the set $\B$ of full blocks of $G$ is computed and traversed in
order of ascending cardinality (beginning with $(S,C)$ such that
$|S\cup C|$ is minimal) in the loop of line~\ref{line:iterate
  blocks}. In the iteration of $(S,C)$, the optimal triangulation of
$\R(S,C)$ is computed.

When processing a block $(S,C)$, the algorithm selects a PMC $\Omega
\in \pmcs(G)$ where $S \subset \Omega \subseteq S \cup C$, to be
saturated according to Equation~\eqref{eq:H by U}, such that the cost
of the resulting triangulation of $\R(S,C)$ is minimized.  The
saturated vertex set is stored as $\Omega(S,C)$
(line~\ref{line:saturated pmc block}). By a slight abuse of notation,
we denote by $\pmcs(S,C)$ the set $\set{\Omega \in \pmcs(G) \mid S
  \subset \Omega \subseteq (S,C)}$.  The chosen optimal triangulation
of $\R(S,C)$ is then stored as $\H(S,C)$ for later use
(line~\ref{line:triang block}). The processing order of the blocks
allows larger blocks to evaluate each member of $\pmcs(S,C)$ based on
previously computed optimal triangulation for each of the realizations
of its smaller blocks. That is, for each block $(S_i,C_i) \in
\fblocks_{\R(S,C)}(\Omega)$, the term $H_i$ in Equation~\eqref{eq:H by
  U} (used in lines \ref{line:saturated pmc block} and
\ref{line:triang block}) will refer to previously computed
$\H(S_i,C_i)$.

Finally, the optimal result is selected by saturating the minimal-cost
PMC in the whole graph (lines~\ref{line:saturated pmc
  graph}--\ref{line:triang graph}).
The following theorem states the correctness and efficiency of the
algorithm. 

\def\ThmAlgIsCorrect{ Let $\kappa$ be a split-monotone bag cost
  computable in polynomial time. $\algname{MinTriang}\angs{\kappa}(G)$
  returns an minimal triangulation of minimal cost $\kappa$ in time
  polynomial in the number of minimal separators of the graph. Hence,
  if $G$ belongs to a poly-MS class of graphs then
  \algname{MinTriang} terminates in polynomial time.}

\begin{theorem}
\label{thm:alg is correct}
\ThmAlgIsCorrect
\end{theorem}

\subsection{Bounded Width}
Another application of the algorithm is when we are interested only in
tree decompositions of a bounded (constant) width $b$, without making
the poly-MS assumption. Bounding the width can be accomplished by
attaching a high cost ($\infty$) to triangulations with maximal
cliques of a larger size than $b$ (lines \ref{line:saturated pmc
  block} and \ref{line:saturated pmc graph}). Furthermore, any minimal
separator larger than $b$ can not be saturated in our output, implying
full blocks of these separators can be completely disregarded in the
main loop (line~\ref{line:iterate blocks}). The limit on the width
bounds the number of minimal separators and PMCs our algorithm should
consider.  Hence, if this limit is considered constant then we get a
polynomial bound on the execution time of our algorithm, without
assuming poly-MS, if we revise line~1 to compute only the minimal
separators and PMCs of size at most $b$. We refer to the revised
algorithm as $\algname{MinTriangB}\angs{b,\kappa}(G)$.  We have the
following.

\begin{theorem}\label{thm:monotone-ptime-b}
  Let $b$ be a fixed natural number, and $\kappa$ a split monotone bag
  cost computable in polynomial time.
  $\algname{MinTriangB}\angs{b,\kappa}(G)$ returns, in polynomial
  time, a minimal triangulation $H$ of width at most $b$ (if one
  exists) with a minimal $\kappa(G,H)$.
\end{theorem}

\section{Enumeration Algorithm}~\label{sec:enum} 
We now present our algorithm for ranked enumeration of minimal
triangulations. A strong connection between minimal triangulations and
proper tree decompositions is stated in
Theorem~\ref{thm:clique-tree-essentials}. Carmeli et
al.~\cite{DBLP:conf/pods/CarmeliKK17} show how enumeration of proper
tree decompositions reduces to that of minimal
triangulations. Formally, we have the following.

\def\propmintraingspropertds{ Let
  $\G$ be class of graphs, and $\kappa$ a bag cost.  If, on graphs of
  $\G$, the minimal triangulations can be enumerated with polynomial
  delay by increasing $\kappa$, then so can the proper tree
  decompositions.  }
\begin{proposition}\label{prop:mintraings-propertds}
\propmintraingspropertds
\end{proposition}
So, in the remainder of this section we consider only minimal
triangulations.  Our algorithm is an application of Lawler-Murty's
procedure~\cite{LAWLER,MURTY}, which reduces ranked enumeration to
optimization under \e{inclusion} and \e{exclusion
  constraints}. The goal of this procedure is to
enumerate sets $A$ of items $a$ by an increasing cost function
$c(A)$. Here, an item $a$ is a minimal separator of $G$ and each set
$A$ is a maximal set of pairwise-parallel minimal separators. Recall
from Theorem~\ref{thm:ParraS97} that each such set $A$ identifies a
minimal triangulation $H$. In particular, the score $c(A)$ is
$\kappa(G,H)$. A constraint of both types is represented as a minimal
separator $S$.  A minimal triangulation $H$ satisfies sets $I$ and $X$
of inclusion and exclusion constraints, respectively, if
$I\subseteq\minseps(H)$ and $X\cap\minseps(H)=\emptyset$. We denote
such a pair as $[I,X]$, and say that $H$ satisfies $[I,X]$ if it
satisfies both $I$ and $X$.

In the rest of this section we will show how $\algname{MinTriang}$
can be adapted to return a minimal triangulation satisfying input constraints,
and present our algorithm for ranked enumeration. Later, we adapt the
algorithm to enumerate the minimal triangulations of a bounded width.

\subsection {Incorporating Constraints}
We would like to restrict $\algname{MinTriang}$
(Figure~\ref{alg:triang}) to return a minimal triangulation that
satisfies a set $[I,X]$ of inclusion and exclusion constraints. Our
approach is to alter our cost function $\kappa$ to assign an infinite
cost to triangulations that violate the constraints. Yet, while doing
so we need to verify that all needed assumptions hold.

To check constraints during the $\algname{MinTriang}$ agorithm we need
to take into consideration two problems that may occure while
triangulating a realization $\R(S,C)$ of $G$. First, $I$ may include
vertices that are not in $\R(S,C)$, so $I$ will be violated for the
wrong reasons. Second, it might be the case that a minimal separator
$S$ of $G$ is not a minimal separator of $\R(S,C)$, but it will be a
minimal separator in a triangulation that contains this
triangulation. Therefore, we use the following equivalent definition
(see Theorem \ref{thm:ParraS97}).  We say that $H$ \e{satisfies}
$[I,X]$, in notation $H \models [I,X]$, if for all $S\in I\cup X$ with
$S\subseteq\nodes(H)$ it holds that $S$ is a clique of $H$ if $S\in I$
and $S$ is \e{not} a clique of $H$ if $S\in X$.



Given a constraint $[I,X]$ and a split-monotone bag cost $\kappa$
(over tree decompositions), we define the cost function $\kappa[I,X]$
as follows.
\begin{equation}
\label{eq:kappa I,X}
\kappa[I,X](G,\T) \eqdef \begin{cases}
\kappa(G,\T) & \mbox{if $H_{\T}\models [I,X]$};\\
\infty & \mbox{otherwise.}
\end{cases}
\end{equation}
where, $H_\T$ is the graph obtained from $G$ by saturating every bag
of $\T$. A crucial lemma is the following, stating that if $\kappa$ is
a split-monotone bag cost, then so is $\kappa[I,X]$.

\def\KappaConstraintsIsOK { Let $\kappa$ be a cost function, $G$ a
  graph, and $[I,X]$ a set of constraints on minimal
  triangulations of $G$.
\begin{cenumerate}
\item If $\kappa$ is a split-monotone bag cost, then so is
  $\kappa[I,X]$.
\item If $\kappa$ can be computed in polynomial time in the size of
  $G$, then $\kappa[I,X]$ can be computed in polynomial time in the
  size of $G$, $I$ and $X$.
\end{cenumerate}
}

\begin{lemma}\label{lemma:monotonic-constraints}
\label{lem:kappa const split monotone}
\KappaConstraintsIsOK
\end{lemma}

Combining Lemma~\ref{lem:kappa const split monotone} with
Theorem~\ref{thm:alg is correct}, we conclude the following theorem
that we will use in the next section.

\begin{theorem}
\label{thm:monotone-constraints-ptime}
Let $\G$ be a poly-MS class of graphs and $\kappa$ a split-monotone
bag cost computable in polynomial time over $\G$.  For all $G\in\G$
and constraints $[I,X]$ over $G$, 
$\algname{MinTriang}\angs{\kappa[I,X]}(G)$ returns a minimum
$\kappa[I,X]$ minimal triangulation in polynomial time.
\end{theorem}

\def\Q{\mathcal{Q}}

\begin{algseries}{b}{Ranked enumeration of the minimal triangulations
    by a split-monotone bag cost $\kappa$.\label{alg:enum}}
\begin{insidealg}{RankedTriang\angs{\kappa}}{$G$}
\STATE $\Q\asn$ empty priority queue by $\kappa$ (lowest first) \label{line:first}
\STATE $H\asn\algname{MinTriang}\angs{\kappa}(G)$\label{line:firstmintriang}
\STATE $\Q.\algname{push}(\angs{H,\emptyset,\emptyset})$ \label{line:first-insert}
\WHILE{$\Q\neq\emptyset$} \label{line:while}
\STATE $\angs{H,I,X}\asn\Q.\algname{pop()}$ \label{line:pop}
\STATE $\mathbf{print}(H)$ \label{line:print}
\STATE let $\minseps(H)\setminus I$ be $\set{S_1,\dots,S_k}$ \label{line:setminus}
\FOR{$i=1,\dots,k-1$}
\STATE $I_i\asn I\cup\set{S_1,\dots,S_{i-1}}$
\STATE $X_i\asn X\cup\set{S_i}$
\STATE $H_i\asn\algname{MinTriang}\angs{\kappa[I_i,X_i]}(G)$ \label{line:secondmintiang}
\IF{$H_i$ satisfies $[I_i,X_i]$} \label{line:if}
\STATE $\Q.\algname{push}(\angs{H_i,I_i,X_i})$ \label{line:secondpush}
\ENDIF
\ENDFOR \label{line:endfor}
\ENDWHILE
\end{insidealg}
\end{algseries}

\subsection{The Enumeration Algorithm}

Our enumeration algorithm, $\algname{RankedTriang}$, is depicted in
Figure~\ref{alg:enum}. It is parameterized by a cost function
$\kappa$, takes as input a graph $G$, and enumerates (via the
$\mathbf{print}$ command of line~\ref{line:print}) the minimal
triangulations of $G$ by increasing cost.

The algorithm maintains a priority queue $\Q$. Each element in $\Q$
represents a partition of the space of minimal triangulations. Here, a
\e{partition} is associated with an inclusion-exclusion constraint
$[I,X]$, and it consists of all minimal triangulations that satisfy
$[I,X]$. In $\Q$, the partition $[I,X]$ is represented as a triple
$\angs{H,I,X}$, where $H$ is a minimum-cost member (minimal
triangulation) in the partition.  Whenever we reach the loop of
line~\ref{line:while}, $\Q$ forms a partition of the entire space of
minimal triangulations that have not been printed yet. In particular,
the first element inserted into $\Q$ (in line~\ref{line:first-insert})
is the triple $\angs{H,\emptyset,\emptyset}$, representing the entire
space of minimal triangulations, where $H$ is a minimal triangulation
of a minimum $\kappa$.

Priority in $\Q$ is determined by the cost of the minimal
triangulation.  In particular, the element removed in
line~\ref{line:pop} is triple $\angs{H,I,X}$ such that
$\kappa(H)\leq\kappa(H')$ for all triple $\angs{H',I',X'}$ in $\Q$. In
each iteration of the while loop, such $\angs{H,I,X}$ is removed and
printed. Then, in the rest of the iteration
(lines~\ref{line:setminus}--\ref{line:endfor}), the remainder of the
partition $[I,X]$ (that is, every minimal triangulation there except
fot $H$) is split into new partitions, all inserted to $Q$
(line~\ref{line:secondpush}). The new partitions are constructed as
follows. Recall that $H$ contains all of the minimal separators in
$I$. Let $S_1,\dots,S_k$ be the set of minimal separators of $H$ that
are \e{not} in $I$. The first partition, $[I_1,X_1]$, is obtained from
$[I,X]$ by copying $I$ and inserting $S_1$ into $X_1$. In the second
partition we add $S_1$ to $I$, but now $X_2$ is $X\cup\set{S_2}$. We
likewise continue for $i=2,\dots i-1$, where each $I_i$ consists of
$I$ and $S_1,\dots,S_{i-1}$, and each $X_i$ is $X\cup\set{S_i}$. It is
an easy observation that the $[I_i,X_i]$ form a proper partition. In
particular, observe that no other minimal triangulation can contain all
of the minimal separators of $H$, since no two triangulations $H$ and $H'$
satisfy $\minseps(H)\subsetneq\minseps(H')$, as implied by
Theorem~\ref{thm:ParraS97}. For each partition $[I_i,X_i]$ a
minimum-cost minimal triangulation $H_i$ w.r.t.~$\kappa[I_i,X_i]$ is
constructed (line~\ref{line:secondmintiang}). Observe that $H_i$
satisfies $[I_i,X_i]$ if and only if this partition is nonempty;
hence, the test of line~\ref{line:if} tests whether the partition is
nonempty, and if so inserts $\angs{H_i,I_i,X_i}$ to $\Q$.

The correctness and efficiency of the algorithm are stated in the
following theorem.

\noamdone{Need some kind of connecting sentence, and maybe a reference
  to the extended version for proofs.}

\begin{theorem}\label{thm:rankedtriang}
  $\algname{RankedTriang}\angs{\kappa}(G)$ enumerates the minimal
  triangulations $H$ of $G$ by increasing $\kappa(G,H)$. Moreover, the
  algorithm enumerates with polynomial delay if:
\begin{cenumerate}
\item $\kappa$ can be computed in polynomial time, \e{and}
\item each $\algname{MinTriang}\angs{\kappa[I_i,X_i]}(G)$
  takes polynomial time.
\end{cenumerate}
\end{theorem}

Combined with Theorem~\ref{thm:monotone-constraints-ptime}, we
conclude the following.

\begin{corollary}\label{thm:rankedtriang-polyms}
  Let $\G$ be a poly-MS class of graphs, and let $\kappa$ be a bag cost
  that is split monotone and computable in polynomial time. On graphs
  $G$ of $\G$, $\algname{RankedTriang}\angs{\kappa}(G)$ can enumerate the
  minimal triangulations of $G$ by increasing $\kappa$, with
  polynomial delay.
\end{corollary}

Note that Theorem~\ref{thm:enum-ranked} (from
Section~\ref{sec:theory}) is a direct consequence of Corollary
\ref{thm:rankedtriang-polyms}.

\partitle{Bounded width}
We now consider the enumeration of the minimal triangulations of a
width bounded by a fixed bound $b$, without masking the assumption of
poly-MS, as stated in Theorem~\ref{thm:enum-ranked-btw}. To adjust
$\algname{RankedTriang}\angs{\kappa}(G)$ to this case, it suffices to
replace $\algname{MinTriang}\angs{\kappa}(G)$ and
$\algname{MinTriang}\angs{\kappa[I_i,X_i]}(G)$, in
lines~\ref{line:firstmintriang} and~\ref{line:secondmintiang}, with
$\algname{MinTriangB}\angs{b,\kappa}(G)$ and
$\algname{MinTriangB}\angs{b,\kappa[I_i,X_i]}(G)$, respectively,
building on Theorem~\ref{thm:monotone-ptime-b} instead of
Theorem~\ref{thm:alg is correct}.  (The algorithm \algname{MinTriangB}
is described right before Theorem~\ref{thm:monotone-ptime-b}.)  We can
show that all guarantees of the enumeration algorithm (as stated in
Theorem~\ref{thm:rankedtriang}) remain valid.
Combined with Lemma~\ref{lemma:monotonic-constraints} and
Proposition~\ref{prop:mintraings-propertds}, we then establish
Theorem~\ref{thm:enum-ranked-btw} (from Section~\ref{sec:theory}).

\def\ckk{\algname{CKK}\xspace}
\def\rankt{\algname{RankedTriang}\xspace}

\section{Experimental Evaluation}~\label{sec:exper} In this section,
we describe our experimental study, where the goal is twofold. First
and foremost, we explore the performance of our enumeration algorithm
(Figures~\ref{alg:triang} and~\ref{alg:enum}), namely
$\algname{RankedTriang}$.  The second goal is to explore the
applicability of the poly-MS assumption in reality, and particularly
to get an insight on how often realistic graphs have a manageable
number of minimal separators.

\subsection{Experimental Setup}

We first describe the general setup for the experiments.

\partitle{Implementation}
The algorithms have been implemeted in C++, with STL data
structures. We used some of the code of Carmeli et
al.~\cite{DBLP:conf/pods/CarmeliKK17}.\footnote{\url{https://github.com/NofarCarmeli/MinTriangulationsEnumeration}}
Particularly, we used their implementation of the algorithm for
enumerating the minimal separators by Berry et
al.~\cite{conf/wg/BerryBC99}. To compute the PMCs of a graph, we
implemented the algorithm by Bouchitt{\'{e}} and
Todinca~\cite{DBLP:journals/tcs/BouchitteT02}. It is important to note
that the implementation of these two algorithms is direct, with no
special optimization. While these algorithms might take a significant
portion of the time, improving their implementation is beyond the
scope of this paper, and remains an important future
direction. \noamdone{Should the next line go here?}We consider the
computation of the minimal separators, PMCs and blocks of each graph
(lines $1$-$2$ of \algname{MinTriang}) the \e{initialization} step, as
they are computed once at the beginning of \algname{RankedTriang},
instead of rerunning for each invocation of \algname{MinTriang} as our
algorithms state.

\partitle{Hardware}
We ran all experiments on a 2.5Ghz 48-core server with 512 GB of RAM
running Ubuntu 14.04.5 LTS. The experiments ran single threaded.\footnote{In
future work we will explore how $\algname{RankedTriang}$ can be
parallelized for delay reduction by parallelizing the main loop or
using more advanced ideas such as those of Golenberg et
al.~\cite{DBLP:journals/pvldb/GolenbergKS11}.}

\partitle{Compared algorithms}
We compared our algorithm to that of Carmeli, Kenig and
Kimelfeld~\cite{DBLP:conf/pods/CarmeliKK17}, referred to as
$\algname{CKK}$, which enumerates with incremental polynomial time and
has no guarantees on the order. To the best of our knowledge, no other
published algorithms for enumerating minimal triangulations or tree
decompositions with completeness guarantees exist, with the exception
of DunceCap~\cite{DBLP:conf/sigmod/TuR15} that is designed for small
query graphs; for more details about its performance we refer the
reader to Carmeli et al.~\cite{DBLP:conf/pods/CarmeliKK17}.
$\algname{CKK}$ requires a black-box minimal triangulator. In our
experiments, we used the algorithm
$\algname{LB\_TRIANG}$~\cite{berry2006wide} for this matter, as it was
found to allow for enumeration of triangulations of smaller width and
fill~\cite{DBLP:conf/pods/CarmeliKK17}.  In principle, we could also
have used our $\algname{MinTriang}$, but we chose not to do so since
$\algname{MinTriang}$ requires a long initialization step, and
\algname{CKK} applies its triangulator to many graphs that change
between execution calls (hence, the initialization cannot be shared
across invocations).

\begin{figure}[t]
\centering
\includegraphics[width=0.45\textwidth]
{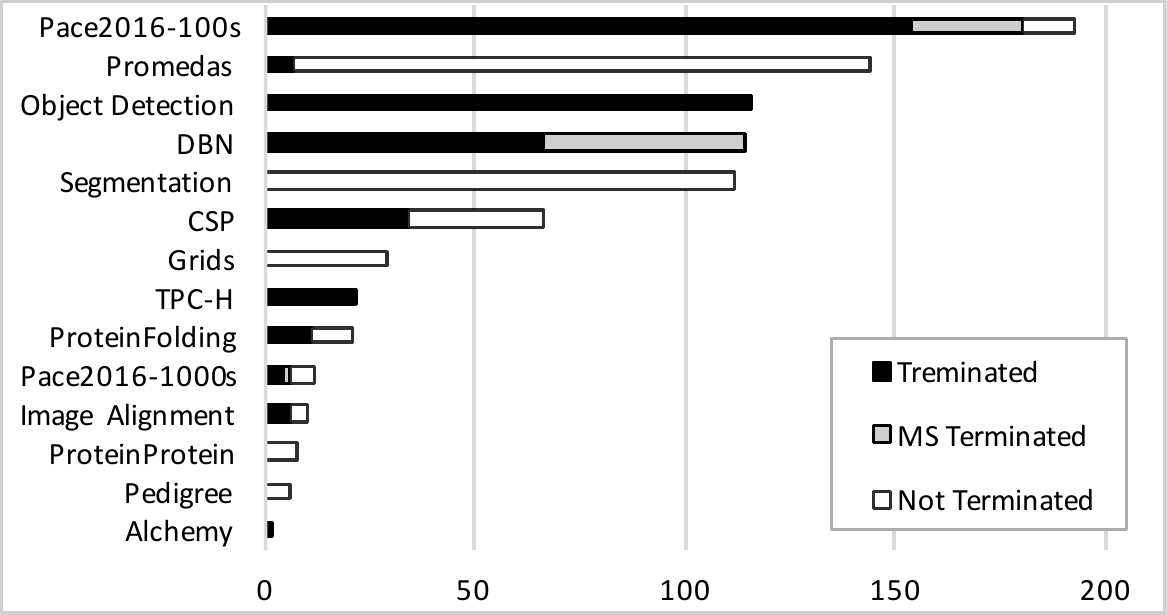}
\caption{\label{fig:distribution}Tractability of computing the minimal
  separators and the PMCs over PIC2011 and PACE2016 graphs.}
\end{figure}

\partitle{Datasets}
Our datasets contain those of Carmeli et
al.~\cite{DBLP:conf/pods/CarmeliKK17}. These include graphs of three
types: probabilistic graphical models from the PIC2011
challenge,\footnote{\url{http://www.cs.huji.ac.il/project/PASCAL/showNet.php}}
Gaifman graphs of conjunctive queries translated from the TPC-H
benchmark (see~\cite{DBLP:conf/pods/CarmeliKK17}), and random graphs.
Random graphs were generated by the $G(n,p)$ Erd\"{o}s-R\'enyi model,
where the number of vertices is $n$ and every pair of vertices is
(independently) connected by an edge with probability $p$.

Additionally, we used the dataset from the PACE2016
competition~\cite{pace2016}, where participants competed on the
computation of tree decompositions (equivalently, minimal
triangulations) of minimal width. These graphs\footnote{The PACE2016
  competition graphs can be found at
  \url{http://github.com/holgerdell/PACE-treewidth-testbed/tree/master/instances/pace16}}
are samples from \e{named
  graphs},\footnote{\url{https://github.com/freetdi/ named-graphs}}
\e{control-flow
  graphs},\footnote{\url{https://github.com/freetdi/CFGs.git}} and the
\e{DIMACS graph-coloring problems}.\footnote{\url{URL:
    http://mat.gsia.cmu.edu/COLOR/instances.  html}} We used graphs
from the tracks where time was limited to 100 and 1000 seconds.

\subsection{The Poly-MS Assumption}
We begin with our investigation of the poly-MS assumption, since this
study is needed as context for the later evaluation of our enumeration
algorithms. In this study, we attempted to generate all minimal
separators, and then all PMCs, on our datasets. We describe the
success rate, and for each successful case the corresponding number of
results.

\partitle{Real-life graphs}
In Figure~\ref{fig:distribution} we report, for each dataset, the
number of graphs where the computation terminated in predefined
time periods. The columns are as follows.
\begin{citemize}
\item \e{Terminated}: Graphs $G$ where the time required to compute
  $\minseps(G)$ is under a minute, and the time required to compute
  $\pmcs(G)$ is under $30$ minutes. These graphs will be used to test
  our algorithm.
\item \e{MS terminated}: Graphs $G$ where the time required to compute
  $\minseps(G)$ is under a minute, but the time to compute $\pmcs(G)$ is
  over $30$ minutes.
\item \e{Not terminated}: Graphs where the time to compute
  $\minseps(G)$ is over a minute.
\end{citemize}

As expected, graphs often violate the poly-MS assumption (otherwise
the NP-hard problem of computing the tree-width and fill-in would
actually be tractable in all of these graphs). In some of the
datasets, all of graphs were found infeasible. Nevertheless, the
positive finding is that the portion of graphs with a manageable
number of minimal separators is quite substantial (around 50\%). The
reader can also observe that in most cases, when we were able to
compute the minimal separators, we were also able to compute the PMCs
(which is consistent with known theory about the relationship between
the two~\cite{DBLP:journals/tcs/BouchitteT02}).
Figure~\ref{fig:tractable} shows the distribution of the number of
minimal separators (in log scale) over the \e{terminated} and \e{MS terminated}
cases of the PIC2011 and PACE2016. One can observe that these numbers are quite
often comparable to the number of edges, and sometimes even smaller.

\begin{figure}[t]
\centering
\includegraphics[width=0.48\textwidth]{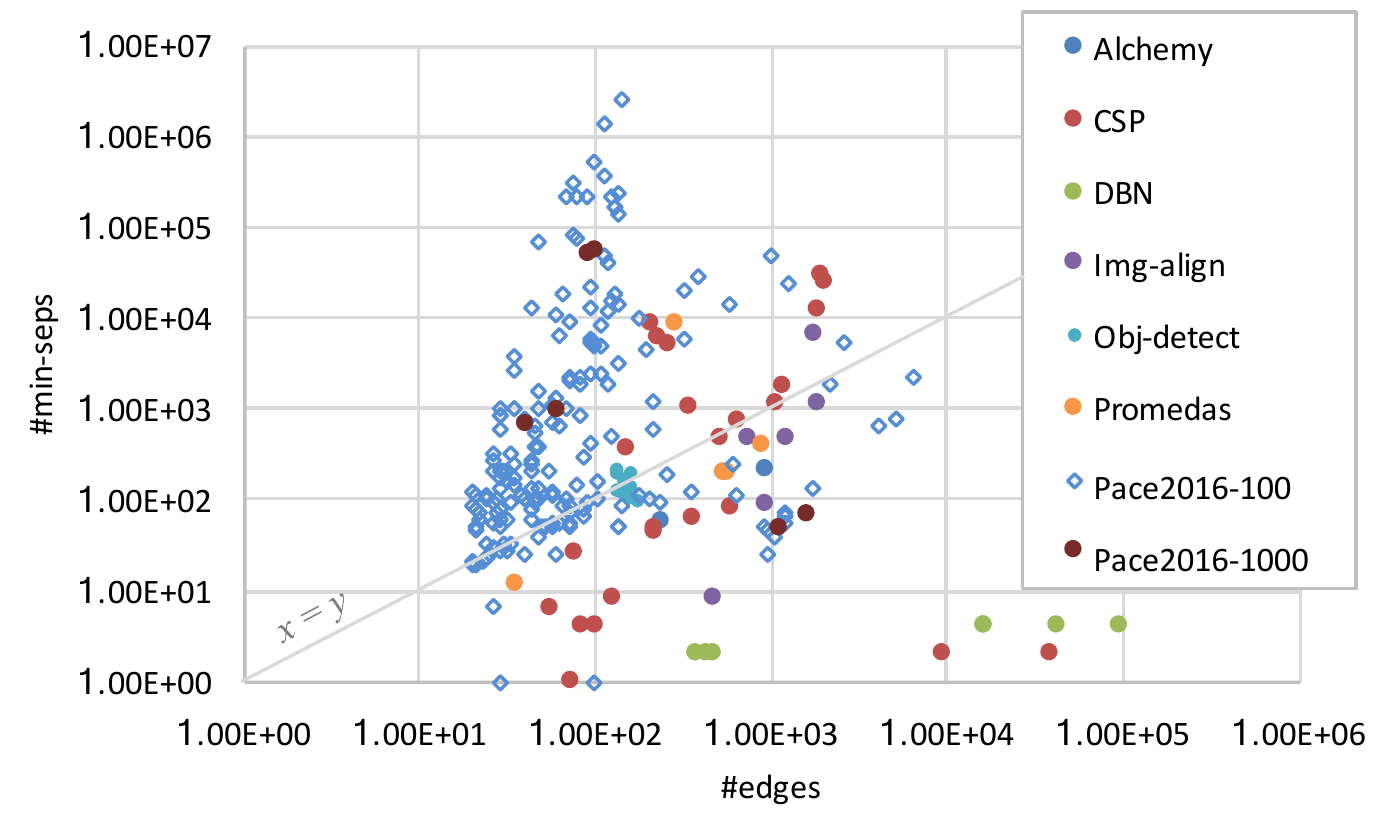}
\caption{\label{fig:tractable}Distribution of the number of minimal
  separators on the MS tractable PIC2011 and PACE2016 graphs.}
\vskip-1em
\end{figure}

\begin{figure}[t]
\centering
\subfigure[$n=20$]{\includegraphics[width=0.235\textwidth]{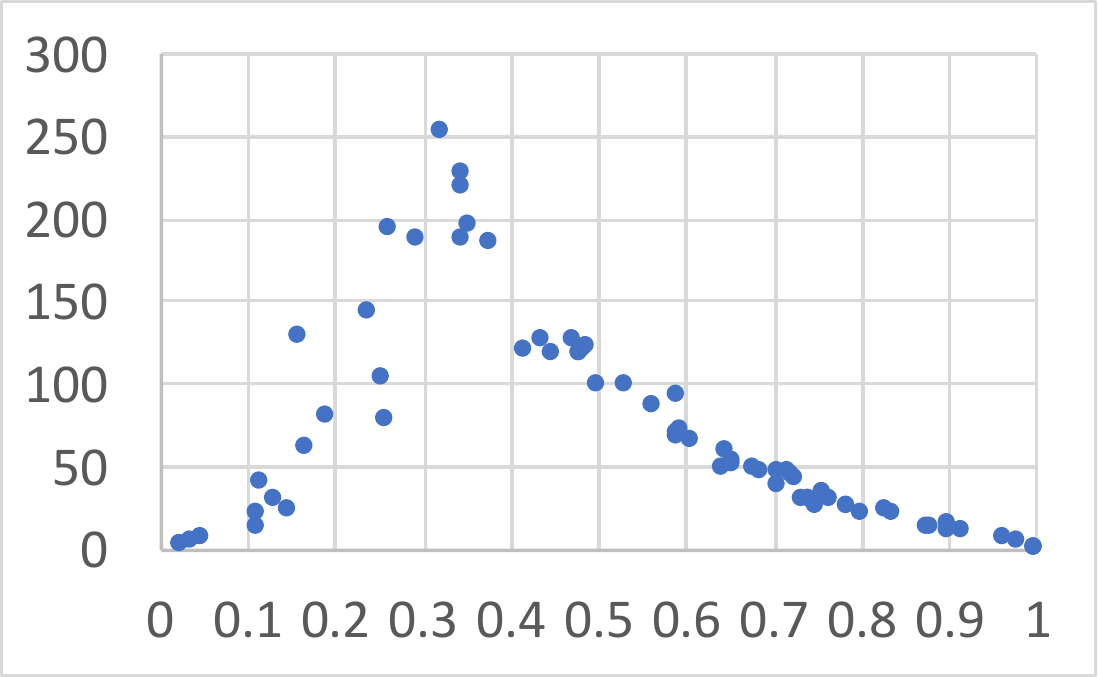}}
\subfigure[$n=30$]{\includegraphics[width=0.235\textwidth]{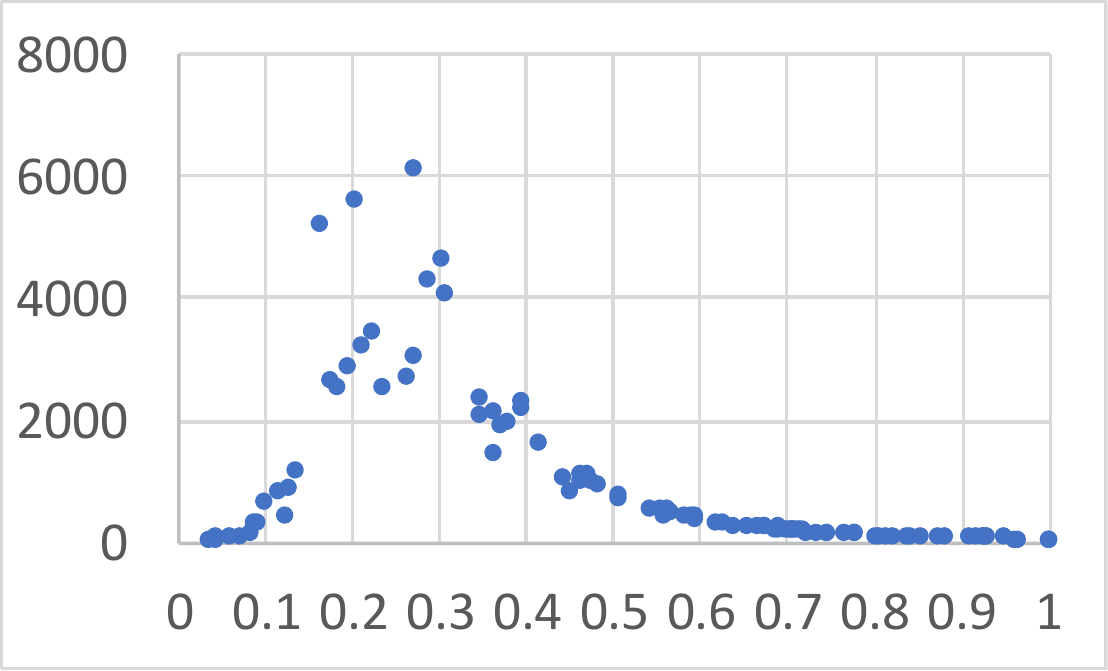}}\\
\subfigure[$n=50$]{\includegraphics[width=0.235\textwidth]{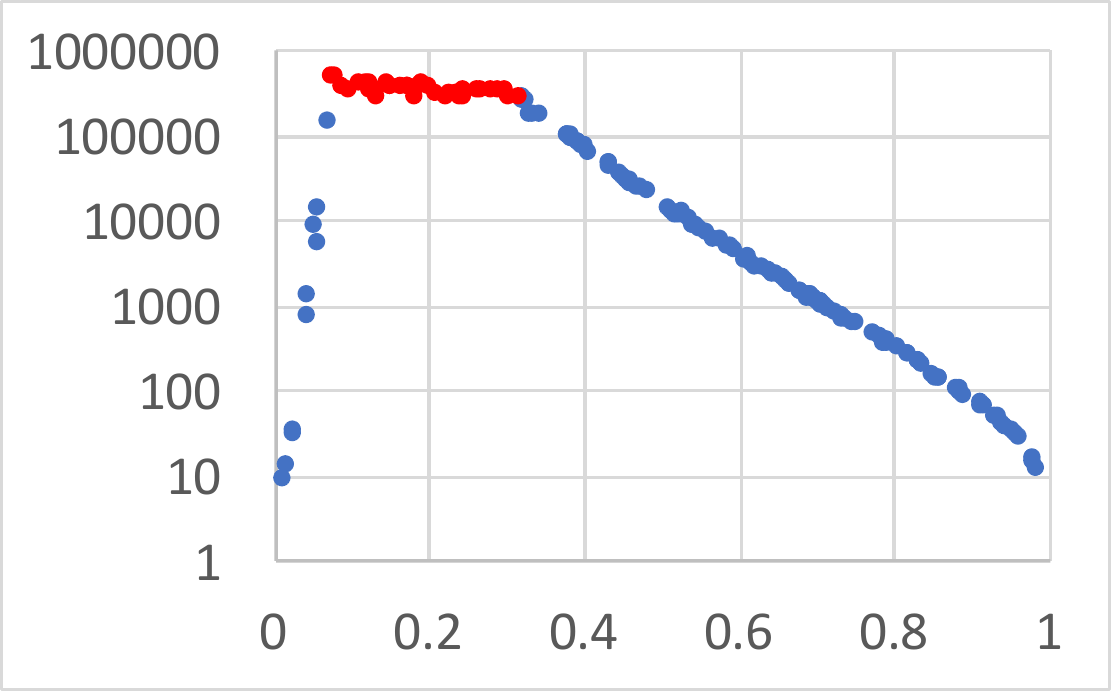}}
\subfigure[$n=70$]{\includegraphics[width=0.235\textwidth]{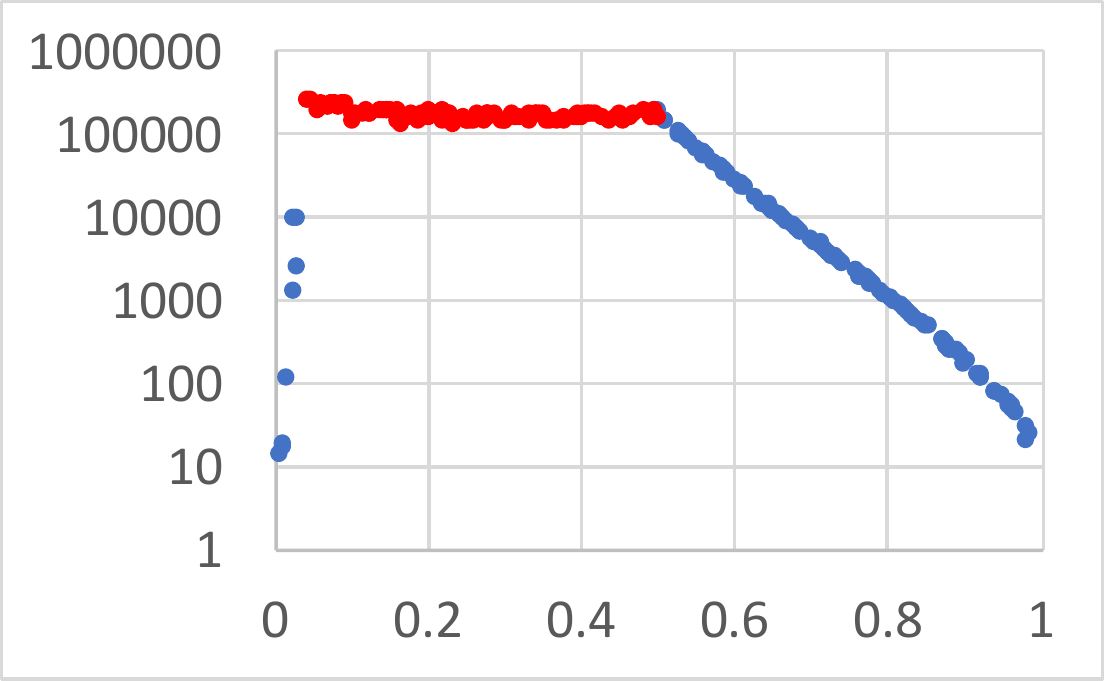}}\\
\caption{\label{fig:dist-random} Number of minimal separators on
  random graph $G(n,p)$. The bottom charts use logarithmic scale.  Red
  marks stand for cases where the computation took over $10$ minutes and
  then stopped. }
\end{figure}
%
\partitle{Random graphs}
\label{sec:exper-random-polyms}
We ran a similar experiment on random graphs.  As said earlier, our
random graphs are $G(n,p)$ from assorted $n$ and $p$. We drew graphs
with $n \in \set{20, 30, 50, 70}$ vertices, drawing three graphs for
each $p \in \set{1/n,\dots,n/n}$. This allowed us
to observe the correlation between the fraction of edges in the graph
and the number of minimal separators. Figure~\ref{fig:dist-random}
reports the result of these tests. When the computation time exceeded
ten minutes, we stopped the execution and, as observed, it happened in
the case of $n=50$ and $n=70$ (shown by the red marks).  The reader
can observe an interesting phenomenon---the number of minimal
separators is small for either sparse or dense graphs. In between
(around $p=0.25$) this number blows up.

{
	\begin{table*}\small
          \caption{Results on 30-minute
            executions, optimizing width and fill. \label{tab:stats} For each dataset, the
            top row relates to
            \algname{RankedTriang} and the bottom to \algname{CKK}. Numbers larger than $1000$ are rounded.}
\vskip-1em
		\label{tab:MCS-M}
		\renewcommand{\arraystretch}{1.1}
		\centering
                  \scalebox{1}{\begin{tabular}
{
|>{\centering}m{2.5cm} 
|c  
|>{\centering}m{0.55cm}  
|c  
|>{\centering}m{0.85cm}  
|>{\centering}m{0.8cm}  
|c 
|c 
|>{\centering}m{0.7cm}  
|>{\centering}m{1.5cm}  
|c  
|}
                      \hline \textbf{dataset (\#graphs)} &
                      \textbf{$\#$trng} & \textbf{init} &
                      \textbf{delay} & \textbf{delay no~init} &
                      \textbf{min-w} & \textbf{$\#$min-w} &
                      \textbf{$\#{\leq}1.1$min-w} &
                      \textbf{min-f} & \textbf{$\#$min-f} &
                      \textbf{$\#{\leq}1.1$min-f}
				\\
				\hline\hline
                                %
				\multirow{2}{*}{CSP $(12)$} &  
				$36493$ & $274.8$ & $2.4$ & $1.7$ & 
				$16.5$ & $34143$ & $34143$ & 
				$146.3$ & $1909$ & $6747$ \\ 
				\cline{2-11}
				&  
				$18625$ & - & $0.23$ & - &  
				$18.5$ & $2117 (12.2\%)$ & $2151 (14\%)$ & 
				$208.2$ & $1 (0.35\%)$ & $336.8 (8.2\%)$ \\
				\hline \hline			
				\multirow{2}{*}{Image Align.~$(4)$} &  
				$29674$ & $36.3$ & $0.36$ & $0.34$ &  
				$20$ & $37067$ & $38951$ &  
				$135$ & $305.5$ & $2640$ \\  
				\cline{2-11}
				&  
				$4191$ & - & $0.57$ & - &  
				$22.3$ & $1034 (1.3\%)$ & $1169 (1.4\%)$ &  
				$213.8$ & $1 (25\%)$ & $10.8 (10.4\%)$ \\
				\hline\hline
				\multirow{2}{*}{\centering Object Detect.~$(79)$} &  
				$538927$ & $0.2$ & $0.004$ & $0.004$ &  
				$5.8$ & $579156$ & $579156$ &   
				$24.9$ & $287589$ & $486465$ \\  
				\cline{2-11}
				&  
				$67639$ & - & $0.03$ & - &   
				$7.3$ & $15946 (4.7\%)$ & $15946 (4.7\%)$ &  
				$39.3$ & $11.9 (0.15\%)$ & $3918 (1.3\%)$ \\  
				\hline \hline
                
                                \multirow{2}{*}{\centering PACE2016-100s~$(87)$} &  
				$137537$ & $108.3$ & $13.4$ & $7.4$ &   
				$5.8$ & $113516$ & $114250$ &    
				$54.4$ & $95139$ & $100612$ \\ 
				\cline{2-11}
				&  
				$164069$ & - & $0.15$ & - &  
				$9.9$ & $87002 (30.3\%)$ & $87214 (30.3\%)$ & 
				$413.4$ & $54781 (767\%)$ & $56994 (904\%)$ \\  
				\hline \hline
				
				\multirow{2}{*}{\centering PACE2016-1000s~$(3)$} &  
				$63029$ & $3.3$ & $0.03$ & $0.03$ & 
				$15.7$ & $36232$ & $39418$ &  
				$56.7$ & $5989$ & $27404$ \\  
				\cline{2-11}
				&  
				$79591$ & - & $0.13$ & - &  
				$16$ & $440.7 (0.4\%)$ & $1363 (8.3\%)$ &  
				$60$ & $13 (3.6\%)$ & $1405 (3.6\%)$ \\  
				\hline \hline
				\multirow{2}{*}{Promedas $(2)$} & 
				$1095$ & $387.6$ & $44.6$ & $28.9$ &  
				$3.5$ & $89.5$ & $89.5$ &  
				$34.5$ & $43$ & $91$ \\
				\cline{2-11}
				%
				%
                                &  
				$10296$ & - & $0.45$ & - & 
				$3.5$ & $9166 (48242\%)$ & $9166 (48242\%)$ &   
				$41.5$ & $0.5 (0.78\%)$ & $0.5 (0.31\%)$ \\
				\hline									
			\end{tabular}}
	\end{table*}
}

\subsection{Enumeration Evaluation}\label{subsec:eval-enumer}
We now describe our evaluation of the algorithm \algname{RankedTriang},
and compare it to \algname{CKK}. 

\partitle{Real-life graphs}
Table~\ref{tab:stats} compares the enumerations of the algorithms on
datasets of PIC2011 and PACE2016, where we were able to compute all
PMCs (i.e., ``Terminated'' graphs from the previous section).  Each
algorithm was executed twice on each graph for $30$ minutes, once for
minimization of width and again for fill-in.  In the case of TPC-H
graphs, computing all minimal triangulations is a matter of a few
seconds, so we did not include those in the table. Moreover, as we
look only at the set of results following a fixed execution time,
\algname{RankedTriang} has no apparent advantage if \algname{CKK}
actually terminates (computing all minimal triangulations); hence,
executions where \algname{CKK} terminates are excluded from the table.
The table columns are as follows.
\begin{itemize}[nosep,leftmargin=1em,labelwidth=*,align=left]
\item \textbf{$\#$trng}: Number of returned minimal
  triangulations.
\item \textbf{init}: For \algname{RankedTriang}, the initialization
  time. Importantly, this time is counted into the $30$
  minutes of the other columns (unless stated otherwise).
\item \textbf{delay}: The average delay between returned results.
\item \textbf{delay no init}: The average delay between returned results, after
initialization.
\item \textbf{min-w}: Minimal width of a minimal triangulation
  returned by the algorithm.
\item \textbf{$\#$min-w}: When the cost is width, number of
  triangulations returned with a minimal width.
\item \textbf{$\#{\leq}1.1\cdot$min-w}: Number of near-optimal (within
  10\%) triangulations returned by the algorithm, optimizing
  width.
\item \textbf{min-f}: Minimal fill-in of a minimal triangulation
  returned by the algorithm.
\item \textbf{$\#$min-f}: When optimizing fill-in, the number of
  triangulations returned with minimal fill-in.
\item \textbf{$\#{\leq}1.1\cdot$min-f}: Number of near optimal triangulations
returned by the algorithm, optimizing fill-in.
\end{itemize} 
Next to the number of optimal (or near optimal) results returned by
\algname{CKK} for each cost, we also report the average percent of optimal 
results returned, relative to \algname{RankedTriang}.

\begin{figure}[t]
\centering
\subfigure[$n=20$\label{fig:20d}]{\includegraphics[width=0.235\textwidth]{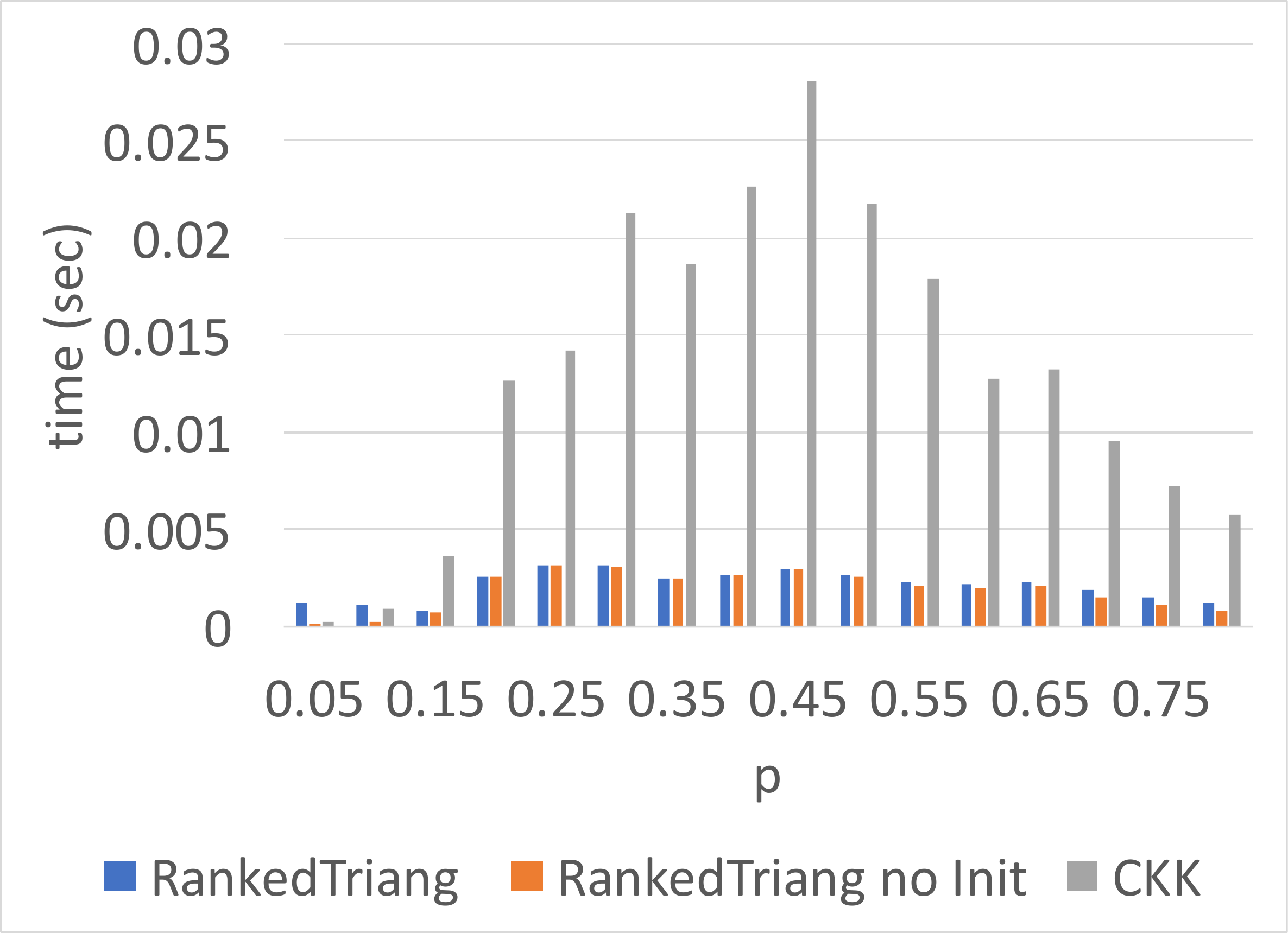}}
\subfigure[$n=50$\label{fig:50d}]{\includegraphics[width=0.235\textwidth]{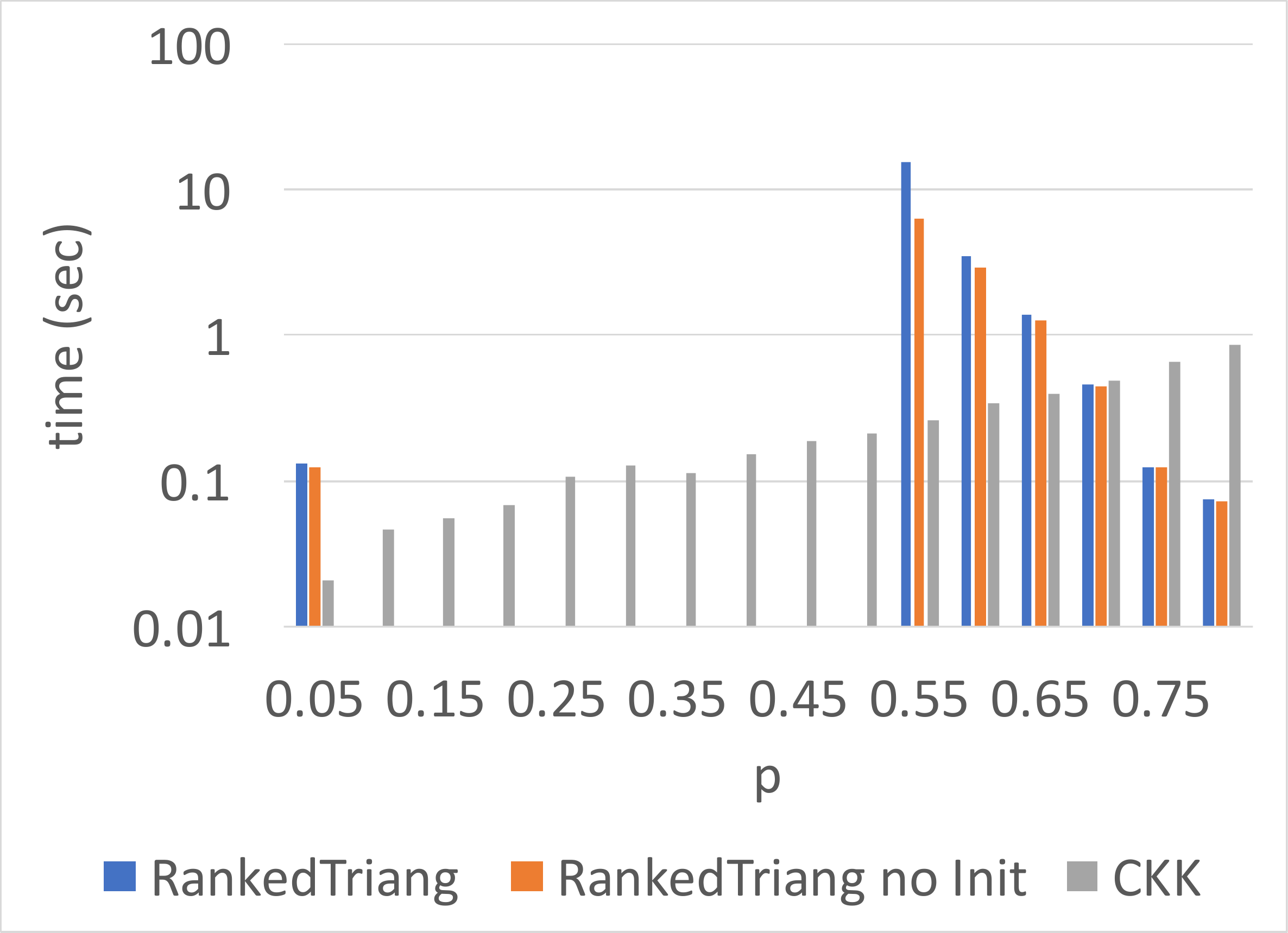}}\\
\subfigure[$n=20$\label{fig:20w}]{\includegraphics[width=0.235\textwidth]{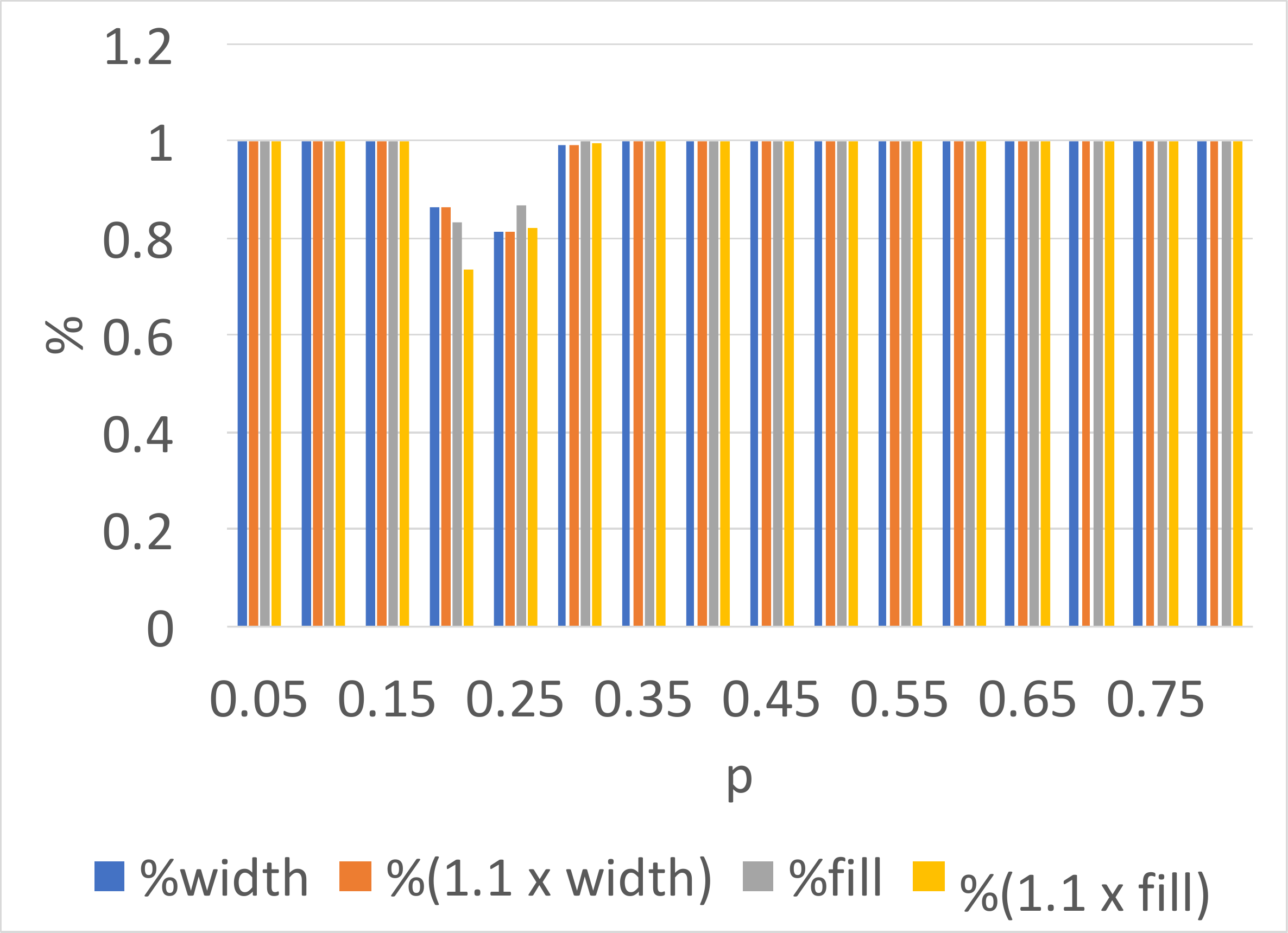}}
\subfigure[$n=50$\label{fig:50w}]{\includegraphics[width=0.235\textwidth]{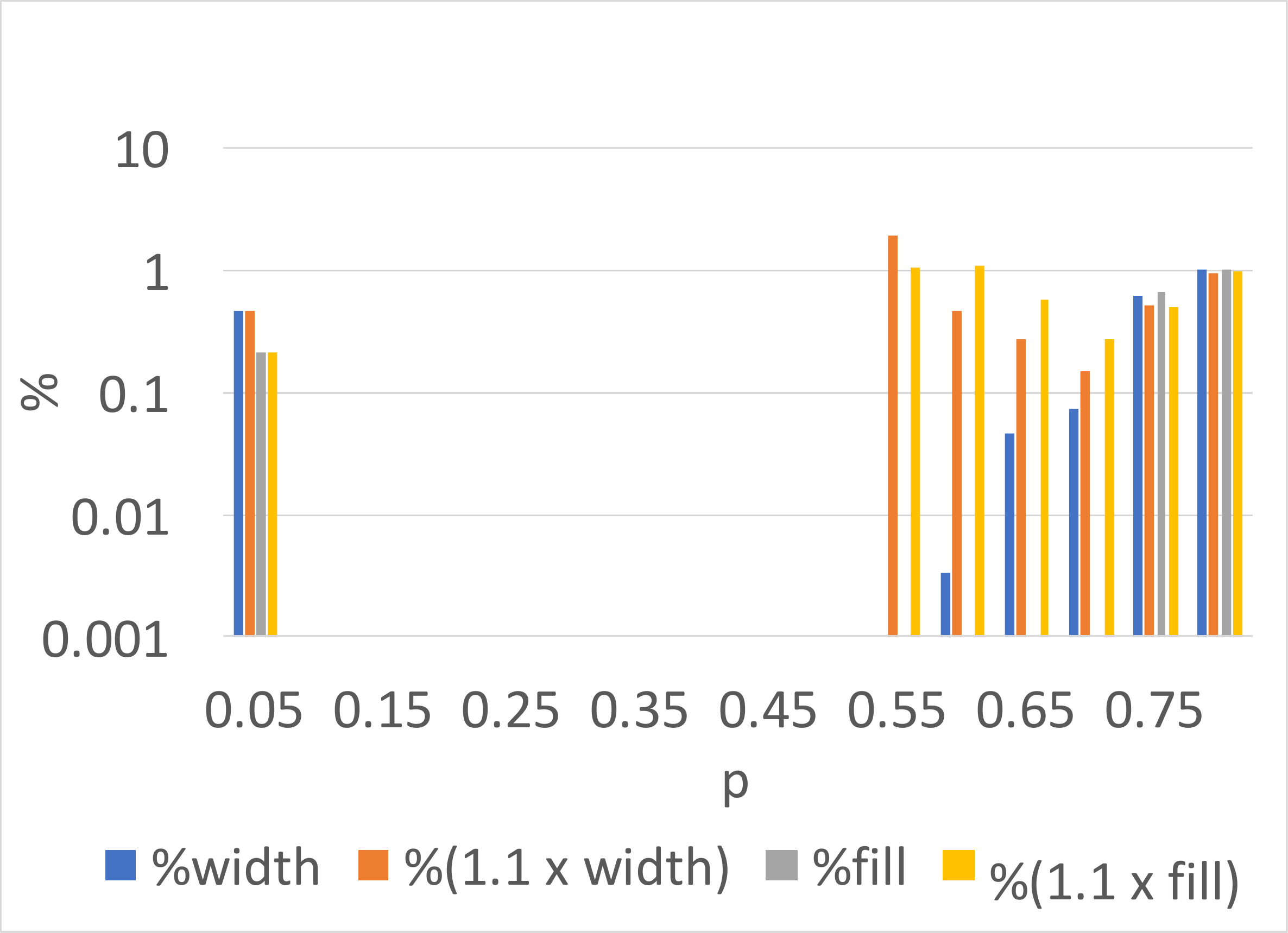}}\\
\caption{\label{fig:dw-random}Delay over random graphs $G(n,p)$,
and ratio of results of \algname{CKK} compared to
\algname{RankedTriang}.
 }
\vskip-1em
\end{figure}

We can see that, with the exception of Promedas, the execution cost of
\algname{RankedTriang} is comparable to, and even lower than,
\algname{CKK}. Moreover, the cost of its answers are consistently
lower than \algname{CKK}, which returns only a fraction of the optimal
triangulations. On Promedas, \algname{RankedTriang} is too slow due to
a high number of PMCs. For the CSP graphs, it appears as if our
algorithm returns more results with a larger average delay.  This is
caused by two graphs where \algname{RankedTriang} returns a higher
number of results than \algname{CKK} by an order of magnitude. This is
not the case for the rest of the dataset, where the delays are longer
and the number of results is smaller in \algname{RankedTriang} than
in \algname{CKK}.

For PACE2016 100s graphs, the table shows the average percentage of
optimal triangulations returned by \algname{CKK} is much larger than
that of \algname{RankedTriang}. This is caused by only three graphs where
the number of results returned by \algname{CKK} is larger by several
orders of magnitude than those returned by \algname{RankedTriang}, due
to a high initialization time and delay. The results returned by the
heuristic $\algname{LB\_TRIANG}$ are all optimal in fill-in for these
three graphs. This is \e{not} the usual case for the graphs in this dataset,
and for many of them \algname{CKK} did not generate \e{any} result
of a minimum fill-in. The reader can observe that under the \e{width}
cost, the advantage of \algname{RankedTriang} over \algname{CKK} is
substantial.


\partitle{Random graphs}
We evaluated both algorithms on random graphs $G(n,p)$, in the same
experimental setup as for the real graphs. We drew graphs with $n \in
\set{20, 50}$ vertices, drawing three graphs for each $p \in
\{0.5,1,1.5,\ldots,0.75,0.8\}$. Figure~\ref{fig:dw-random} reports the
results. Figures~\ref{fig:20d} and~\ref{fig:50d} show the average
delay between returned results, where the delay of
\algname{RankedTriang} is measured with and without the initialization
time.  Figures~\ref{fig:20w} and~\ref{fig:50w} show the relative
percent of optimal cost results returned by \algname{CKK}, relative to
\algname{RankedTriang}. We can see that for graphs where the minimal
separators and PMCs can be computed, the delay of
\algname{RankedTriang} is smaller than that of \algname{CKK}. This
occures in all graphs with 20 vertices, and some of the graphs with 50
vertices.  In the case of graphs with 50 vertices,
\algname{RankedTriang} does not terminate initialization for
probabilities $0.1$ through $0.5$, implying for these graphs the
number of minimal separators and PMCs is high.  For graphs from
probabilities closer to this range, \algname{RankedTriang} has a
higher delay, which is caused by a larger number of minimal
separators.  These results are consistent with the hardness results
observed in Section~\ref{sec:exper-random-polyms}.

In Appendix~\ref{app:case}, we describe case studies on two specific
graphs.


\section{Concluding Remarks}~\label{sec:conclusions} We presented an
algorithm for enumerating minimal triangulations, and by implication
proper tree decompositions, by increasing cost with polynomial delay,
for arbitrary split-monotone (efficiently computable) bag costs. One
variant of the algorithm enumerates \e{all} minimal triangulations,
but requires pre-computing all minimal separators, and hence, makes
the poly-MS assumption. The other variant computes all minimal
triangulations of a bounded width, and does not require the poly-MS
assumption. Our implementation and experimental study shows that the
algorithm can lend itself to practical realization, and that the
poly-MS assumption is quite often valid.

Various directions are left open for future research, beyond the ones
mentioned throughout the paper. For one, it is known that computing a
tree decomposition of width $w$ is \e{Fixed-Parameter Tractable}
(FPT)~\cite{DBLP:journals/siamcomp/Bodlaender96} when taking $w$ as
the parameter; this means that $w$ affects only the constant (and not
the degree) of the polynomial~\cite{DBLP:series/mcs/DowneyF99}. Can we
extend this result to the enumeration of all tree decompositions of
width at most $w$ with FPT delay? Also, can we strengthen our
algorithms with further \e{diversity} of results to maximize the
potential value to the application? How should diversification be
defined? In addition to investigating the open questions, we plan to
incorporate our algorithm in the framework of Abseher et
al.~\cite{DBLP:journals/jair/AbseherMW17} for supervised learning of
effective tree decompositions.


\eat{
\section*{Acknowledgments}
The authors are very grateful to Nofar Carmeli and Batya Kenig for
insightful comments and discussions.
}

{
\scriptsize
\bibliographystyle{ACM-Reference-Format}
\bibliography{ms}


\begin{thebibliography}{42}


\ifx \showCODEN    \undefined \def \showCODEN     #1{\unskip}     \fi
\ifx \showDOI      \undefined \def \showDOI       #1{#1}\fi
\ifx \showISBNx    \undefined \def \showISBNx     #1{\unskip}     \fi
\ifx \showISBNxiii \undefined \def \showISBNxiii  #1{\unskip}     \fi
\ifx \showISSN     \undefined \def \showISSN      #1{\unskip}     \fi
\ifx \showLCCN     \undefined \def \showLCCN      #1{\unskip}     \fi
\ifx \shownote     \undefined \def \shownote      #1{#1}          \fi
\ifx \showarticletitle \undefined \def \showarticletitle #1{#1}   \fi
\ifx \showURL      \undefined \def \showURL       {\relax}        \fi
\providecommand\bibfield[2]{#2}
\providecommand\bibinfo[2]{#2}
\providecommand\natexlab[1]{#1}
\providecommand\showeprint[2][]{arXiv:#2}

\bibitem[\protect\citeauthoryear{Abseher, Musliu, and Woltran}{Abseher
  et~al\mbox{.}}{2017}]%
        {DBLP:journals/jair/AbseherMW17}
\bibfield{author}{\bibinfo{person}{Michael Abseher}, \bibinfo{person}{Nysret
  Musliu}, {and} \bibinfo{person}{Stefan Woltran}.}
  \bibinfo{year}{2017}\natexlab{}.
\newblock \showarticletitle{Improving the Efficiency of Dynamic Programming on
  Tree Decompositions via Machine Learning}.
\newblock \bibinfo{journal}{\emph{J. Artif. Intell. Res.}}
  \bibinfo{volume}{58} (\bibinfo{year}{2017}), \bibinfo{pages}{829--858}.
\newblock


\bibitem[\protect\citeauthoryear{Berry, Blair, and Heggernes}{Berry
  et~al\mbox{.}}{2002}]%
        {Berry:2002:MCS:647683.732496}
\bibfield{author}{\bibinfo{person}{Anne Berry}, \bibinfo{person}{Jean R.~S.
  Blair}, {and} \bibinfo{person}{Pinar Heggernes}.}
  \bibinfo{year}{2002}\natexlab{}.
\newblock \showarticletitle{Maximum Cardinality Search for Computing Minimal
  Triangulations}. In \bibinfo{booktitle}{\emph{WG}} \emph{(\bibinfo{series}{WG
  '02})}. \bibinfo{publisher}{Springer-Verlag}, \bibinfo{address}{London, UK,
  UK}, \bibinfo{pages}{1--12}.
\newblock
\showISBNx{3-540-00331-2}


\bibitem[\protect\citeauthoryear{Berry, Bordat, and Cogis}{Berry
  et~al\mbox{.}}{1999}]%
        {conf/wg/BerryBC99}
\bibfield{author}{\bibinfo{person}{Anne Berry}, \bibinfo{person}{Jean~Paul
  Bordat}, {and} \bibinfo{person}{Olivier Cogis}.}
  \bibinfo{year}{1999}\natexlab{}.
\newblock \showarticletitle{Generating All the Minimal Separators of a Graph.}.
  In \bibinfo{booktitle}{\emph{WG}} \emph{(\bibinfo{series}{lncs})},
  Vol.~\bibinfo{volume}{1665}. \bibinfo{publisher}{Springer},
  \bibinfo{pages}{167--172}.
\newblock


\bibitem[\protect\citeauthoryear{Berry, Bordat, Heggernes, Simonet, and
  Villanger}{Berry et~al\mbox{.}}{2006a}]%
        {BERRY200633}
\bibfield{author}{\bibinfo{person}{Anne Berry}, \bibinfo{person}{Jean-Paul
  Bordat}, \bibinfo{person}{Pinar Heggernes}, \bibinfo{person}{Geneviève
  Simonet}, {and} \bibinfo{person}{Yngve Villanger}.}
  \bibinfo{year}{2006}\natexlab{a}.
\newblock \showarticletitle{A wide-range algorithm for minimal triangulation
  from an arbitrary ordering}.
\newblock \bibinfo{journal}{\emph{Journal of Algorithms}} \bibinfo{volume}{58},
  \bibinfo{number}{1} (\bibinfo{year}{2006}), \bibinfo{pages}{33 -- 66}.
\newblock


\bibitem[\protect\citeauthoryear{Berry, Bordat, Heggernes, Simonet, and
  Villanger}{Berry et~al\mbox{.}}{2006b}]%
        {berry2006wide}
\bibfield{author}{\bibinfo{person}{Anne Berry}, \bibinfo{person}{Jean-Paul
  Bordat}, \bibinfo{person}{Pinar Heggernes}, \bibinfo{person}{Genevi{\`e}ve
  Simonet}, {and} \bibinfo{person}{Yngve Villanger}.}
  \bibinfo{year}{2006}\natexlab{b}.
\newblock \showarticletitle{A wide-range algorithm for minimal triangulation
  from an arbitrary ordering}.
\newblock \bibinfo{journal}{\emph{Journal of Algorithms}} \bibinfo{volume}{58},
  \bibinfo{number}{1} (\bibinfo{year}{2006}), \bibinfo{pages}{33--66}.
\newblock


\bibitem[\protect\citeauthoryear{Blair and Peyton}{Blair and Peyton}{1993}]%
        {ChordalIntroduction}
\bibfield{author}{\bibinfo{person}{Jean~RS Blair} {and} \bibinfo{person}{Barry
  Peyton}.} \bibinfo{year}{1993}\natexlab{}.
\newblock \showarticletitle{An introduction to chordal graphs and clique
  trees}.
\newblock In \bibinfo{booktitle}{\emph{Graph theory and sparse matrix
  computation}}. \bibinfo{publisher}{Springer}, \bibinfo{pages}{1--29}.
\newblock


\bibitem[\protect\citeauthoryear{Bodlaender}{Bodlaender}{1996}]%
        {DBLP:journals/siamcomp/Bodlaender96}
\bibfield{author}{\bibinfo{person}{Hans~L. Bodlaender}.}
  \bibinfo{year}{1996}\natexlab{}.
\newblock \showarticletitle{A Linear-Time Algorithm for Finding
  Tree-Decompositions of Small Treewidth}.
\newblock \bibinfo{journal}{\emph{{SIAM} J. Comput.}} \bibinfo{volume}{25},
  \bibinfo{number}{6} (\bibinfo{year}{1996}), \bibinfo{pages}{1305--1317}.
\newblock


\bibitem[\protect\citeauthoryear{Bouchitt{\'e} and Todinca}{Bouchitt{\'e} and
  Todinca}{2001}]%
        {mintriangalg}
\bibfield{author}{\bibinfo{person}{Vincent Bouchitt{\'e}} {and}
  \bibinfo{person}{Ioan Todinca}.} \bibinfo{year}{2001}\natexlab{}.
\newblock \showarticletitle{Treewidth and minimum fill-in: Grouping the minimal
  separators}.
\newblock \bibinfo{journal}{\emph{SIAM J. Comput.}} \bibinfo{volume}{31},
  \bibinfo{number}{1} (\bibinfo{year}{2001}), \bibinfo{pages}{212--232}.
\newblock


\bibitem[\protect\citeauthoryear{Bouchitt{\'{e}} and Todinca}{Bouchitt{\'{e}}
  and Todinca}{2002}]%
        {DBLP:journals/tcs/BouchitteT02}
\bibfield{author}{\bibinfo{person}{Vincent Bouchitt{\'{e}}} {and}
  \bibinfo{person}{Ioan Todinca}.} \bibinfo{year}{2002}\natexlab{}.
\newblock \showarticletitle{Listing all potential maximal cliques of a graph}.
\newblock \bibinfo{journal}{\emph{Theor. Comput. Sci.}} \bibinfo{volume}{276},
  \bibinfo{number}{1-2} (\bibinfo{year}{2002}), \bibinfo{pages}{17--32}.
\newblock


\bibitem[\protect\citeauthoryear{Carmeli, Kenig, and Kimelfeld}{Carmeli
  et~al\mbox{.}}{2017}]%
        {DBLP:conf/pods/CarmeliKK17}
\bibfield{author}{\bibinfo{person}{Nofar Carmeli}, \bibinfo{person}{Batya
  Kenig}, {and} \bibinfo{person}{Benny Kimelfeld}.}
  \bibinfo{year}{2017}\natexlab{}.
\newblock \showarticletitle{Efficiently Enumerating Minimal Triangulations}. In
  \bibinfo{booktitle}{\emph{PODS}}. \bibinfo{publisher}{{ACM}},
  \bibinfo{pages}{273--287}.
\newblock


\bibitem[\protect\citeauthoryear{Cohen, Kimelfeld, and Sagiv}{Cohen
  et~al\mbox{.}}{2008}]%
        {DBLP:journals/jcss/CohenKS08}
\bibfield{author}{\bibinfo{person}{Sara Cohen}, \bibinfo{person}{Benny
  Kimelfeld}, {and} \bibinfo{person}{Yehoshua Sagiv}.}
  \bibinfo{year}{2008}\natexlab{}.
\newblock \showarticletitle{Generating all maximal induced subgraphs for
  hereditary and connected-hereditary graph properties}.
\newblock \bibinfo{journal}{\emph{J. Comput. Syst. Sci.}} \bibinfo{volume}{74},
  \bibinfo{number}{7} (\bibinfo{year}{2008}), \bibinfo{pages}{1147--1159}.
\newblock


\bibitem[\protect\citeauthoryear{Dell, Husfeldt, Jansen, Kaski, Komusiewicz,
  and Rosamond}{Dell et~al\mbox{.}}{2017}]%
        {pace2016}
\bibfield{author}{\bibinfo{person}{Holger Dell}, \bibinfo{person}{Thore
  Husfeldt}, \bibinfo{person}{Bart M.~P. Jansen}, \bibinfo{person}{Petteri
  Kaski}, \bibinfo{person}{Christian Komusiewicz}, {and}
  \bibinfo{person}{Frances~A. Rosamond}.} \bibinfo{year}{2017}\natexlab{}.
\newblock \showarticletitle{{The First Parameterized Algorithms and
  Computational Experiments Challenge}}. In \bibinfo{booktitle}{\emph{IPEC}}
  \emph{(\bibinfo{series}{Leibniz International Proceedings in Informatics
  (LIPIcs)})}, Vol.~\bibinfo{volume}{63}. \bibinfo{pages}{30:1--30:9}.
\newblock
\showISBNx{978-3-95977-023-1}
\showISSN{1868-8969}


\bibitem[\protect\citeauthoryear{Downey and Fellows}{Downey and
  Fellows}{1999}]%
        {DBLP:series/mcs/DowneyF99}
\bibfield{author}{\bibinfo{person}{Rodney~G. Downey} {and}
  \bibinfo{person}{Michael~R. Fellows}.} \bibinfo{year}{1999}\natexlab{}.
\newblock \bibinfo{booktitle}{\emph{Parameterized Complexity}}.
\newblock \bibinfo{publisher}{Springer}.
\newblock


\bibitem[\protect\citeauthoryear{Fomin, Todinca, and Villanger}{Fomin
  et~al\mbox{.}}{2015}]%
        {DBLP:journals/siamcomp/FominTV15}
\bibfield{author}{\bibinfo{person}{Fedor~V. Fomin}, \bibinfo{person}{Ioan
  Todinca}, {and} \bibinfo{person}{Yngve Villanger}.}
  \bibinfo{year}{2015}\natexlab{}.
\newblock \showarticletitle{Large Induced Subgraphs via Triangulations and
  {CMSO}}.
\newblock \bibinfo{journal}{\emph{{SIAM} J. Comput.}} \bibinfo{volume}{44},
  \bibinfo{number}{1} (\bibinfo{year}{2015}), \bibinfo{pages}{54--87}.
\newblock


\bibitem[\protect\citeauthoryear{Furuse and Yamazaki}{Furuse and
  Yamazaki}{2014}]%
        {DBLP:journals/tcs/FuruseY14}
\bibfield{author}{\bibinfo{person}{Masanobu Furuse} {and}
  \bibinfo{person}{Koichi Yamazaki}.} \bibinfo{year}{2014}\natexlab{}.
\newblock \showarticletitle{A revisit of the scheme for computing treewidth and
  minimum fill-in}.
\newblock \bibinfo{journal}{\emph{Theor. Comput. Sci.}}  \bibinfo{volume}{531}
  (\bibinfo{year}{2014}), \bibinfo{pages}{66--76}.
\newblock


\bibitem[\protect\citeauthoryear{Golenberg, Kimelfeld, and Sagiv}{Golenberg
  et~al\mbox{.}}{2011}]%
        {DBLP:journals/pvldb/GolenbergKS11}
\bibfield{author}{\bibinfo{person}{Konstantin Golenberg},
  \bibinfo{person}{Benny Kimelfeld}, {and} \bibinfo{person}{Yehoshua Sagiv}.}
  \bibinfo{year}{2011}\natexlab{}.
\newblock \showarticletitle{Optimizing and Parallelizing Ranked Enumeration}.
\newblock \bibinfo{journal}{\emph{{PVLDB}}} \bibinfo{volume}{4},
  \bibinfo{number}{11} (\bibinfo{year}{2011}), \bibinfo{pages}{1028--1039}.
\newblock


\bibitem[\protect\citeauthoryear{Gottlob, Greco, and Scarcello}{Gottlob
  et~al\mbox{.}}{2005a}]%
        {DBLP:journals/jair/GottlobGS05}
\bibfield{author}{\bibinfo{person}{Georg Gottlob}, \bibinfo{person}{Gianluigi
  Greco}, {and} \bibinfo{person}{Francesco Scarcello}.}
  \bibinfo{year}{2005}\natexlab{a}.
\newblock \showarticletitle{Pure Nash Equilibria: Hard and Easy Games}.
\newblock \bibinfo{journal}{\emph{J. Artif. Intell. Res. {(JAIR)}}}
  \bibinfo{volume}{24} (\bibinfo{year}{2005}), \bibinfo{pages}{357--406}.
\newblock


\bibitem[\protect\citeauthoryear{Gottlob, Grohe, Musliu, Samer, and
  Scarcello}{Gottlob et~al\mbox{.}}{2005b}]%
        {DBLP:conf/wg/GottlobGMSS05}
\bibfield{author}{\bibinfo{person}{Georg Gottlob}, \bibinfo{person}{Martin
  Grohe}, \bibinfo{person}{Nysret Musliu}, \bibinfo{person}{Marko Samer}, {and}
  \bibinfo{person}{Francesco Scarcello}.} \bibinfo{year}{2005}\natexlab{b}.
\newblock \showarticletitle{Hypertree Decompositions: Structure, Algorithms,
  and Applications}. In \bibinfo{booktitle}{\emph{{WG}}}
  \emph{(\bibinfo{series}{lncs})}, Vol.~\bibinfo{volume}{3787}.
  \bibinfo{publisher}{Springer}, \bibinfo{pages}{1--15}.
\newblock


\bibitem[\protect\citeauthoryear{Gottlob, Leone, and Scarcello}{Gottlob
  et~al\mbox{.}}{1999}]%
        {DBLP:conf/pods/GottlobLS99}
\bibfield{author}{\bibinfo{person}{Georg Gottlob}, \bibinfo{person}{Nicola
  Leone}, {and} \bibinfo{person}{Francesco Scarcello}.}
  \bibinfo{year}{1999}\natexlab{}.
\newblock \showarticletitle{Hypertree Decompositions and Tractable Queries}. In
  \bibinfo{booktitle}{\emph{PODS}}. \bibinfo{publisher}{{ACM} Press},
  \bibinfo{pages}{21--32}.
\newblock


\bibitem[\protect\citeauthoryear{Gottlob, Leone, and Scarcello}{Gottlob
  et~al\mbox{.}}{2002}]%
        {GOTTLOB2002579}
\bibfield{author}{\bibinfo{person}{Georg Gottlob}, \bibinfo{person}{Nicola
  Leone}, {and} \bibinfo{person}{Francesco Scarcello}.}
  \bibinfo{year}{2002}\natexlab{}.
\newblock \showarticletitle{Hypertree Decompositions and Tractable Queries}.
\newblock \bibinfo{journal}{\emph{J. Comput. System Sci.}}
  \bibinfo{volume}{64}, \bibinfo{number}{3} (\bibinfo{year}{2002}),
  \bibinfo{pages}{579 -- 627}.
\newblock
\showISSN{0022-0000}


\bibitem[\protect\citeauthoryear{Gottlob, Mikl{\'{o}}s, and Schwentick}{Gottlob
  et~al\mbox{.}}{2009}]%
        {DBLP:journals/jacm/GottlobMS09}
\bibfield{author}{\bibinfo{person}{Georg Gottlob},
  \bibinfo{person}{Zolt{\'{a}}n Mikl{\'{o}}s}, {and} \bibinfo{person}{Thomas
  Schwentick}.} \bibinfo{year}{2009}\natexlab{}.
\newblock \showarticletitle{Generalized hypertree decompositions: {NP}-hardness
  and tractable variants}.
\newblock \bibinfo{journal}{\emph{J. {ACM}}} \bibinfo{volume}{56},
  \bibinfo{number}{6} (\bibinfo{year}{2009}).
\newblock


\bibitem[\protect\citeauthoryear{Grohe and Marx}{Grohe and Marx}{2014}]%
        {DBLP:journals/talg/GroheM14}
\bibfield{author}{\bibinfo{person}{Martin Grohe} {and}
  \bibinfo{person}{D{\'{a}}niel Marx}.} \bibinfo{year}{2014}\natexlab{}.
\newblock \showarticletitle{Constraint Solving via Fractional Edge Covers}.
\newblock \bibinfo{journal}{\emph{{ACM} Trans. Algorithms}}
  \bibinfo{volume}{11}, \bibinfo{number}{1} (\bibinfo{year}{2014}),
  \bibinfo{pages}{4:1--4:20}.
\newblock


\bibitem[\protect\citeauthoryear{Johnson, Papadimitriou, and
  Yannakakis}{Johnson et~al\mbox{.}}{1988}]%
        {DBLP:journals/ipl/JohnsonP88}
\bibfield{author}{\bibinfo{person}{David~S. Johnson},
  \bibinfo{person}{Christos~H. Papadimitriou}, {and} \bibinfo{person}{Mihalis
  Yannakakis}.} \bibinfo{year}{1988}\natexlab{}.
\newblock \showarticletitle{On Generating All Maximal Independent Sets}.
\newblock \bibinfo{journal}{\emph{Inf. Process. Lett.}} \bibinfo{volume}{27},
  \bibinfo{number}{3} (\bibinfo{year}{1988}), \bibinfo{pages}{119--123}.
\newblock


\bibitem[\protect\citeauthoryear{Jordan}{Jordan}{2002}]%
        {Jordan}
\bibfield{author}{\bibinfo{person}{M. Jordan}.}
  \bibinfo{year}{2002}\natexlab{}.
\newblock \bibinfo{booktitle}{\emph{An Introduction to Probabilistic Graphical
  Models}}.
\newblock \bibinfo{publisher}{University of California, Berkeley}, Chapter~17.
\newblock


\bibitem[\protect\citeauthoryear{Kalinsky, Etsion, and Kimelfeld}{Kalinsky
  et~al\mbox{.}}{2017}]%
        {DBLP:conf/edbt/KalinskyEK17}
\bibfield{author}{\bibinfo{person}{Oren Kalinsky}, \bibinfo{person}{Yoav
  Etsion}, {and} \bibinfo{person}{Benny Kimelfeld}.}
  \bibinfo{year}{2017}\natexlab{}.
\newblock \showarticletitle{Flexible Caching in Trie Joins}. In
  \bibinfo{booktitle}{\emph{EDBT}}. \bibinfo{publisher}{OpenProceedings.org},
  \bibinfo{pages}{282--293}.
\newblock


\bibitem[\protect\citeauthoryear{Kenig and Gal}{Kenig and Gal}{2015}]%
        {DBLP:conf/sum/KenigG15}
\bibfield{author}{\bibinfo{person}{Batya Kenig} {and} \bibinfo{person}{Avigdor
  Gal}.} \bibinfo{year}{2015}\natexlab{}.
\newblock \showarticletitle{On the Impact of Junction-Tree Topology on Weighted
  Model Counting}. In \bibinfo{booktitle}{\emph{{SUM}}}
  \emph{(\bibinfo{series}{lncs})}, Vol.~\bibinfo{volume}{9310}.
  \bibinfo{publisher}{Springer}, \bibinfo{pages}{83--98}.
\newblock


\bibitem[\protect\citeauthoryear{Kloks, Kratsch, and Spinrad}{Kloks
  et~al\mbox{.}}{1997}]%
        {DBLP:journals/tcs/KloksKS97}
\bibfield{author}{\bibinfo{person}{Ton Kloks}, \bibinfo{person}{Dieter
  Kratsch}, {and} \bibinfo{person}{Jeremy~P. Spinrad}.}
  \bibinfo{year}{1997}\natexlab{}.
\newblock \showarticletitle{On Treewidth and Minimum Fill-In of Asteroidal
  Triple-Free Graphs}.
\newblock \bibinfo{journal}{\emph{Theor. Comput. Sci.}} \bibinfo{volume}{175},
  \bibinfo{number}{2} (\bibinfo{year}{1997}), \bibinfo{pages}{309--335}.
\newblock


\bibitem[\protect\citeauthoryear{Kolaitis and Vardi}{Kolaitis and
  Vardi}{2000}]%
        {DBLP:journals/jcss/KolaitisV00}
\bibfield{author}{\bibinfo{person}{Phokion~G. Kolaitis} {and}
  \bibinfo{person}{Moshe~Y. Vardi}.} \bibinfo{year}{2000}\natexlab{}.
\newblock \showarticletitle{Conjunctive-Query Containment and Constraint
  Satisfaction}.
\newblock \bibinfo{journal}{\emph{J. Comput. Syst. Sci.}} \bibinfo{volume}{61},
  \bibinfo{number}{2} (\bibinfo{year}{2000}), \bibinfo{pages}{302--332}.
\newblock


\bibitem[\protect\citeauthoryear{Lauritzen and Spiegelhalter}{Lauritzen and
  Spiegelhalter}{1988}]%
        {LauSpi-JRS88}
\bibfield{author}{\bibinfo{person}{S. Lauritzen} {and} \bibinfo{person}{D.~J.
  Spiegelhalter}.} \bibinfo{year}{1988}\natexlab{}.
\newblock \showarticletitle{Local Computations with Probabilities on Graphical
  Structures and Their Application to Expert Systems}.
\newblock \bibinfo{journal}{\emph{Journal of the Royal Statistical Society}}
  \bibinfo{volume}{B, 50}, \bibinfo{number}{2} (\bibinfo{year}{1988}),
  \bibinfo{pages}{157--224}.
\newblock


\bibitem[\protect\citeauthoryear{Lawler}{Lawler}{1972}]%
        {LAWLER}
\bibfield{author}{\bibinfo{person}{E.~L. Lawler}.}
  \bibinfo{year}{1972}\natexlab{}.
\newblock \showarticletitle{A procedure for computing the $K$ best solutions to
  discrete optimization problems and its application to the shortest path
  problem}.
\newblock \bibinfo{journal}{\emph{Management Science}} \bibinfo{volume}{18},
  \bibinfo{number}{7} (\bibinfo{year}{1972}), \bibinfo{pages}{401--405}.
\newblock


\bibitem[\protect\citeauthoryear{Liedloff, Montealegre, and Todinca}{Liedloff
  et~al\mbox{.}}{2015}]%
        {DBLP:conf/wg/LiedloffMT15}
\bibfield{author}{\bibinfo{person}{Mathieu Liedloff}, \bibinfo{person}{Pedro
  Montealegre}, {and} \bibinfo{person}{Ioan Todinca}.}
  \bibinfo{year}{2015}\natexlab{}.
\newblock \showarticletitle{Beyond Classes of Graphs with "Few" Minimal
  Separators: {FPT} Results Through Potential Maximal Cliques}. In
  \bibinfo{booktitle}{\emph{WG}} \emph{(\bibinfo{series}{lncs})},
  Vol.~\bibinfo{volume}{9224}. \bibinfo{publisher}{Springer},
  \bibinfo{pages}{499--512}.
\newblock


\bibitem[\protect\citeauthoryear{Marx}{Marx}{2010}]%
        {DBLP:journals/talg/Marx10}
\bibfield{author}{\bibinfo{person}{D{\'{a}}niel Marx}.}
  \bibinfo{year}{2010}\natexlab{}.
\newblock \showarticletitle{Approximating fractional hypertree width}.
\newblock \bibinfo{journal}{\emph{{ACM} Trans. Algorithms}}
  \bibinfo{volume}{6}, \bibinfo{number}{2} (\bibinfo{year}{2010}).
\newblock


\bibitem[\protect\citeauthoryear{Mediero}{Mediero}{2017}]%
        {mediero2017search}
\bibfield{author}{\bibinfo{person}{H{\'e}ctor~Otero Mediero}.}
  \bibinfo{year}{2017}\natexlab{}.
\newblock \bibinfo{booktitle}{\emph{Search Algorithms for Solving Queries on
  Graphical Models and the Importance of Pseudo-trees in their Complexity}}.
\newblock \bibinfo{type}{{T}echnical {R}eport}.
  \bibinfo{institution}{University of California Irvine}.
\newblock
\newblock
\shownote{UCI ICS Technical Report.}


\bibitem[\protect\citeauthoryear{Montealegre and Todinca}{Montealegre and
  Todinca}{2016}]%
        {DBLP:conf/wg/MontealegreT16}
\bibfield{author}{\bibinfo{person}{Pedro Montealegre} {and}
  \bibinfo{person}{Ioan Todinca}.} \bibinfo{year}{2016}\natexlab{}.
\newblock \showarticletitle{On Distance-d Independent Set and Other Problems in
  Graphs with "few" Minimal Separators}. In \bibinfo{booktitle}{\emph{WG}}
  \emph{(\bibinfo{series}{lncs})}, Vol.~\bibinfo{volume}{9941}.
  \bibinfo{pages}{183--194}.
\newblock


\bibitem[\protect\citeauthoryear{Murty}{Murty}{1968}]%
        {MURTY}
\bibfield{author}{\bibinfo{person}{K.~G. Murty}.}
  \bibinfo{year}{1968}\natexlab{}.
\newblock \showarticletitle{An algorithm for ranking all the assignments in
  order of increasing cost}.
\newblock \bibinfo{journal}{\emph{Operations Research}} \bibinfo{volume}{16},
  \bibinfo{number}{3} (\bibinfo{year}{1968}), \bibinfo{pages}{682--687}.
\newblock


\bibitem[\protect\citeauthoryear{Otachi and Schweitzer}{Otachi and
  Schweitzer}{2014}]%
        {DBLP:conf/swat/OtachiS14}
\bibfield{author}{\bibinfo{person}{Yota Otachi} {and} \bibinfo{person}{Pascal
  Schweitzer}.} \bibinfo{year}{2014}\natexlab{}.
\newblock \showarticletitle{Reduction Techniques for Graph Isomorphism in the
  Context of Width Parameters}. In \bibinfo{booktitle}{\emph{SWAT}}
  \emph{(\bibinfo{series}{lncs})}, Vol.~\bibinfo{volume}{8503}.
  \bibinfo{publisher}{Springer}, \bibinfo{pages}{368--379}.
\newblock


\bibitem[\protect\citeauthoryear{Parra and Scheffler}{Parra and
  Scheffler}{1997}]%
        {DBLP:journals/dam/ParraS97}
\bibfield{author}{\bibinfo{person}{Andreas Parra} {and} \bibinfo{person}{Petra
  Scheffler}.} \bibinfo{year}{1997}\natexlab{}.
\newblock \showarticletitle{Characterizations and Algorithmic Applications of
  Chordal Graph Embeddings}.
\newblock \bibinfo{journal}{\emph{Discrete Applied Mathematics}}
  \bibinfo{volume}{79}, \bibinfo{number}{1-3} (\bibinfo{year}{1997}),
  \bibinfo{pages}{171--188}.
\newblock


\bibitem[\protect\citeauthoryear{Rose}{Rose}{1970}]%
        {rose1970triangulated}
\bibfield{author}{\bibinfo{person}{Donald~J Rose}.}
  \bibinfo{year}{1970}\natexlab{}.
\newblock \showarticletitle{Triangulated graphs and the elimination process}.
\newblock \bibinfo{journal}{\emph{J. Math. Anal. Appl.}} \bibinfo{volume}{32},
  \bibinfo{number}{3} (\bibinfo{year}{1970}), \bibinfo{pages}{597--609}.
\newblock


\bibitem[\protect\citeauthoryear{Tarjan and Yannakakis}{Tarjan and
  Yannakakis}{1984}]%
        {Tarjan:1984:SLA:1169.1179}
\bibfield{author}{\bibinfo{person}{Robert~E. Tarjan} {and}
  \bibinfo{person}{Mihalis Yannakakis}.} \bibinfo{year}{1984}\natexlab{}.
\newblock \showarticletitle{Simple Linear-time Algorithms to Test Chordality of
  Graphs, Test Acyclicity of Hypergraphs, and Selectively Reduce Acyclic
  Hypergraphs}.
\newblock \bibinfo{journal}{\emph{SIAM J. Comput.}} \bibinfo{volume}{13},
  \bibinfo{number}{3} (\bibinfo{date}{July} \bibinfo{year}{1984}),
  \bibinfo{pages}{566--579}.
\newblock


\bibitem[\protect\citeauthoryear{Tu and R{\'{e}}}{Tu and R{\'{e}}}{2015}]%
        {DBLP:conf/sigmod/TuR15}
\bibfield{author}{\bibinfo{person}{Susan Tu} {and} \bibinfo{person}{Christopher
  R{\'{e}}}.} \bibinfo{year}{2015}\natexlab{}.
\newblock \showarticletitle{{DunceCap}: Query Plans Using Generalized Hypertree
  Decompositions}. In \bibinfo{booktitle}{\emph{SIGMOD}}.
  \bibinfo{publisher}{{ACM}}, \bibinfo{pages}{2077--2078}.
\newblock


\bibitem[\protect\citeauthoryear{Yamada, Kataoka, and Watanabe}{Yamada
  et~al\mbox{.}}{2010}]%
        {DBLP:journals/ijcm/YamadaKW10}
\bibfield{author}{\bibinfo{person}{Takeo Yamada}, \bibinfo{person}{Seiji
  Kataoka}, {and} \bibinfo{person}{Kohtaro Watanabe}.}
  \bibinfo{year}{2010}\natexlab{}.
\newblock \showarticletitle{Listing all the minimum spanning trees in an
  undirected graph}.
\newblock \bibinfo{journal}{\emph{Int. J. Comput. Math.}} \bibinfo{volume}{87},
  \bibinfo{number}{14} (\bibinfo{year}{2010}), \bibinfo{pages}{3175--3185}.
\newblock


\bibitem[\protect\citeauthoryear{Zhao, Malmberg, and Cai}{Zhao
  et~al\mbox{.}}{2006}]%
        {DBLP:conf/wabi/ZhaoMC06}
\bibfield{author}{\bibinfo{person}{Jizhen Zhao}, \bibinfo{person}{Russell~L.
  Malmberg}, {and} \bibinfo{person}{Liming Cai}.}
  \bibinfo{year}{2006}\natexlab{}.
\newblock \showarticletitle{Rapid \emph{ab initio} {RNA} Folding Including
  Pseudoknots Via Graph Tree Decomposition}. In
  \bibinfo{booktitle}{\emph{WABI}}. \bibinfo{pages}{262--273}.
\newblock


\end{thebibliography}
}

\newpage
\appendix

\section{Additional Proofs}

\subsection{Proof of Theorem~\ref{thm:alg is correct}}

In this section, we will prove our algorithm returns a minimal triangulation of
optimal cost when $\kappa$ is a split monotone bag cost. Since $\kappa$ is a
bag cost, its value is equal for all clique trees of a triangulation, and we
can prove our lemma by focusing on one clique tree of the minimal triangulation.

Let $H$ be a minimal triangulation of a graph $G$ and $\Omega \in
\maxcliques(H)$. We identify a certain tree decomposition of $H$,
and show it is a clique tree. Our definition is recursive over the blocks of
$\Omega$, in an approach similar to the triangulation algorithm. For each block
$(S_i,C_i) \in \fblocks_G(\Omega)$ we note $\T_i = (T_i,\beta_i)$ a clique tree
of $H_i = \induced{H}{S_i \cup C_i}$. Note $\Omega_i \in \maxcliques(H_i)$ such
that $S_i \subset \Omega_i$ (at least one exists by Theorem
\ref{theorem:realization triangulation}) and $v_i \in \nodes(T_i)$ such that
$\beta_i(v_i) = \Omega_i$.

We denote $\T_H = (T_H, \beta_H)$ a tree decomposition that is
a union of all $\T_i$, connected by a node $v_\Omega$
representing $\Omega$. $\T_H$ is defined as follows:
\begin{align*}
\nodes(T_H) =& \bigcup_{i=1}^p \nodes(T_i) \cup \{v_\Omega\}\\
\edges(T_H) =& \bigcup_{i=1}^p \Big( \edges(T_i) \cup
\{\{v_\Omega,v_i\}\} \Big) \\
\beta_H(u) =& \begin{cases}
  \Omega & u = v_\Omega \\
  \beta_{H_i}(u) & u \in \nodes(T_i)
\end{cases}
\end{align*}

To prove $\T_H$ is a clique tree, we first characterize the maximal cliques of a
minimal triangulation $H$, in the same recursive manner.

\begin{lemma}
\label{lem:cliques of bestH}
Let $H$ be a minimal triangulation of the graph $G$. Let $\Omega$ be a maximal
clique of $H$ and note $H_i = \induced{H}{S_i \cup C_i}$ for each block
$(S_i,C_i) \in \fblocks_G(\Omega)$ where $1 \leq i \leq p$ for $p=|\fblocks_G(\Omega)|$.
Then:
\[
\maxcliques(H) = \bigcup_{i = 1}^p \maxcliques(H_i) \cup \big\{\Omega\big\}
\]
\end{lemma}

\begin{proof}
Theorem \ref{theorem:pmc-block-triangulation} directly implies:
\begin{equation}
\label{eq:H of maximal clique}
H = \bigcup_{i=1}^p H_i \cup \Omega^*
\end{equation} 
$\K_\Omega$ and each $H_i$ are all subgraphs of $H$, and by Theorem
\ref{theorem:pmc-block-triangulation} the graph $\bigcup_{i=1}^p H_i \cup
\K_\Omega$ is a minimal triangulation of $G$. Therefore, if $H$ is a minimal
triangulation it must hold that $H = \bigcup_{i=1}^p H_i \cup \K_\Omega$.

From equation \ref{eq:H of maximal clique} we can conclude the maximal cliques
of $H$ by observing that for any block $(S_i,C_i) \in \fblocks(\Omega)$ the
maximal cliques of $H_i$ are maximal cliques in $H$. This is intuitively true,
as the vertices of $C_i$ are not connected to new nodes, and by Theorem
\ref{theorem:realization triangulation} $S_i$ can not be a maximal clique of
$H_i$, implying all maximal cliques of $H_i$ are maximal cliques in $H$.
\end{proof}

Due to Lemma \ref{lem:cliques of bestH}, we can now deduce that $\T_H$ is a
clique tree of $H$.

\begin{lemma}
\label{lem:clique tree of bestH}
$\T_h$ is a clique tree of $H$.
\end{lemma}

\begin{proof}
To prove $\T_H$ is a clique tree of $H$, we must prove:
\begin{enumerate}
  \item The vertices of $H$ are covered
  \item The edges of $H$ are covered
  \item The junction-tree property holds
  \item $\beta_H$ is a bijection between $\nodes(T_H)$ and $\maxcliques(H)$
\end{enumerate}

Properties $(1),(2)$ and $(4)$ are directly implied from the definition of
$\T_H$ and lemma \ref{lem:cliques of bestH}. We will now show that $\T_H$
upholds the junction-tree property - for each two nodes $u,v,w \in
\nodes(T_H)$ where $w$ is a vertex on the route from $u$ to $v$ in $T_H$, it
holds that $\beta(u) \cap \beta(v) \subseteq \beta(w)$.

We can assume there exists a block $(S_i,C_i) \in \fblocks_G(\Omega)$ such that
$u \in \nodes(T_i)$. If $v \in \nodes(T_i)$, the path between $u$ and $v$
exists in the clique tree $\T_i$, $w \in \nodes(T_i)$, and the junction-tree
property must hold. 

Otherwise, note $v_i \in \nodes(T_i)$ the node connected to $v_\Omega$ in the
construction of $\T_H$. Since $\beta(v)$ is not contained in the block
$(S_i,C_i)$, and by definition of $v_i$ the following must hold:
\[
\beta(u) \cap \beta(v) \subseteq S_i = \beta(v_i) \cap \Omega
\]
This implies if $\beta(w) = \Omega$, then $\beta(u) \cap \beta(v) \subseteq
\beta(w)$ and the junction-tree property holds. Furthermore, if $w \in
\nodes(T_i)$, it must be on the path from $v_i$ to $u$. According to our
assumption, since the clique intersection property holds in $\T_i$ it
also holds in $\T_H$:
\[
\beta(u) \cap \beta(v) \subseteq S_i \subseteq \beta(v_i) \cap \beta(u)
\subseteq \beta(w)
\]

Finally, if there exists another block $(S_j,C_j) \in \fblocks_G(\Omega)$ such
that $w \in \nodes(T_j)$, due to the structure of $\T_H$ we know $v
\in \nodes(T_j)$. In this case, we can switch the roles of $u$ and $v$,
and we have already seen when $w$ belongs to the same block as $u$ the clique
intersection property holds.
\end{proof}

Next, we would like to clarify the implications of a split monotone bag cost,
when used as a cost function over triangulations. $\T_H$ is used to
prove these implications.

\begin{lemma}
\label{lem:split monotone cost implication}
Let $H$ and $H'$ be two minimal triangulations of a graph $G$, $\Omega$
be a maximal clique in $\maxcliques(H) \cap \maxcliques(H')$ and $\kappa$ 
a split monotone bag cost over tree decompositions. Suppose there is a block
$(S,C) \in \fblocks_G(\Omega)$ such that:
\begin{align*}
\induced{H}{\nodes(G) \setminus C} &= \induced{H'}{\nodes(G) \setminus C} \\ 
\kappa(\induced{G}{(S,C)}, \induced{H}{(S,C)}) &\leq
\kappa(\induced{G}{(S,C)}, \induced{H'}{(S,C)})
\end{align*}
then $\kappa(G, H) \leq \kappa(G, H')$.
\end{lemma}
\begin{proof}
Note $H_1 = \induced{H}{(S,C)}$ and $H_2 = \induced{H'}{(S,C)}$. 

Observe the clique tree $\T_H = (T_H,\beta_H)$. It must have a pair of nodes
$u, v_\Omega \in \nodes(T_H)$ such that $\beta_H(v_\Omega) = \Omega$,
$\beta_H(u) \in \maxcliques(H_1)$, the edge $e=\{u,v_\Omega\} \in \edges(T_H)$
and $\beta_H(u) \cap \beta_H(v_\Omega) = S$.
$\T_H$ splits by $e$ as $\angs{\induced{G}{\nodes(G) \setminus
C},\T_1,\induced{G}{(S,C)},\T_2}$ where $\T_1$ and $\T_2$ are clique trees of
$\induced{H}{\nodes(G) \setminus C}$ and $H_1$, respectively. Similarly,
$\T_{H'}$ can be split by an edge $e'$ as $\angs{\induced{G}{\nodes(G)\setminus
C}, \T'_1, \induced{G}{(S,C)},\T'_2}$ where $\T'_1$ and $\T'_2$ are clique trees
of $\induced{H'}{\nodes(G) \setminus C}$ and $H_2$, respectively.

Since $\induced{H'}{\nodes(G)\setminus C} = \induced{H}{\nodes(G)\setminus C}$
the clique trees $\T_1$ and $\T'_1$ have the same bags. As $\kappa$ is a bag
cost it holds that:
\[
\kappa(\induced{G}{\nodes(G)\setminus C}, \T_1) =
\kappa(\induced{G}{\nodes(G)\setminus C}, \T'_1)
\]
By assumption:
\begin{align*}
&\kappa(\induced{G}{(S,C)}, T_2) = \kappa(\induced{G}{(S,C)}, H_1) <
\\
& 
\kappa(\induced{G}{(S,C)}, H_2) = \kappa(\induced{G}{(S,C)}, T'_2)
\end{align*}
$\kappa$ is a split monotone function and therefore:
\[
\kappa(G, H)=\kappa(G,\T_H) \leq \kappa(G,\T_{H'})=\kappa(G, H')
\]
as claimed.
\end{proof}

Finally, we can prove the correctness of our algorithm.

\begin{replemma}{\ref{thm:alg is correct}}
\ThmAlgIsCorrect
\end{replemma}



\begin{proof}
First, we will prove that for each minimal separator $S \in \minseps(G)$ and
full $S$-component $C$, $\H(S,C)$ (as calculated in line \ref{line:triang block}
of $\algname{MinTriang}$) is an optimal minimal triangulation of the realization
$\R(S,C)$, with the cost $\kappa(\induced{G}{(S,C)}, \H(S,C))$. We assume in
contradiction this is not the case, and let $(S,C)$ be the smallest block in $G$
such that there exists a minimal triangulation $H'$ of the realization $\R(S,C)$
where $H' \neq \H(S,C)$ and $\kappa(\induced{G}{(S,C)},H') <
\kappa(\induced{G}{(S,C)},\H(S,C))$.

Notice, since $H'$ is a minimal triangulation different from $\H(S,C)$, it
can not be a single clique. According to theorem \ref{theorem:realization
triangulation}, there exists a maximal clique $\Omega'$ in $H'$ such that $S
\subset \Omega' \subset (S, C)$. For each full block $(S_i,C_i) \in
\fblocks_{\R(S,C)}(\Omega')$, we note $H'_i = \induced{H'}{S_i \cup C_i}$, and
by assumption the minimal triangulation $\H(S_i,C_i)$ is optimal (as it is of
smaller cardinality than $(S,C)$), and the following holds:
\begin{align*}
\kappa(\induced{G}{(S_i,C_i)},H'_i) &\geq
\kappa(\induced{G}{(S_i,C_i)},\H(S_i,C_i)) \\
&= \kappa(\induced{G}{(S_i,C_i)},\induced{\HByPMC{\R(S,C)}{\Omega'}}{S_i \cup
C_i})
\end{align*}
Since $\kappa$ is a split monotone bag cost we can use lemma \ref{lem:split
monotone cost implication} and by selection of $\Omega(S,C)$ (line
\ref{line:saturated pmc block} in $\algname{MinTriang}$) this implies:
\begin{align*}
\kappa(\induced{G}{(S,C)},H')  
&\geq \kappa(\induced{G}{(S,C)},\HByPMC{\R(S,C)}{\Omega'}) \\
&\geq \kappa(\induced{G}{(S,C)},\HByPMC{\R(S,C)}{\Omega(S,C)})\\
&= \kappa(\induced{G}{(S,C)},\H{(S,C)})
\end{align*}
in contradiction to our assumption.

We can conclude the optimality of $\H(G)$ in a similar manner, based on
theorem \ref{theorem:pmc-block-triangulation}. Let $\Omega'$ be any
maximal clique of $H'$. As previously, the following inequality holds:
\begin{align*}
\kappa(G,H')  
&\geq \kappa(G,\HByPMC{G}{\Omega'}) \\
&\geq \kappa(G, \HByPMC{G}{\Omega(G)})\\
&= \kappa(G, \H(G))
\end{align*}
This concludes $\H(G)$ is an optimal minimal triangulation of $G$.
\end{proof}

\subsection{Proof of Lemma \ref{lem:kappa const split monotone}}
\begin{replemma}{\ref{lem:kappa const split monotone}}
\KappaConstraintsIsOK
\end{replemma}

\begin{proof}
Let $H,H'$ be two minimal triangulations of a graph $G$ and $\Omega$ be a
maximal clique in $\maxcliques(H) \cap \maxcliques(H')$. Let $\T$ and $\T'$ be
tree decompositions of $H$ and $H'$, respectively, such that $\T$ and
$\T'$ can be split as $\angs{G_1,\T_1,G_2,\T_2}$ and
$\angs{G_1,\T_1',G_2,\T_2'}$, respectively. We assume (for $i=1,2$):
\begin{equation*}
\kappa[I,X](G_i,\T_i) \leq \kappa[I,X](G_i,\T_i')
\end{equation*}
We would like to prove:
\begin{equation}
\label{eq:monotony-of-constaint-kappa}
\kappa[I,X](G,\T) \leq \kappa[I,X](G,\T')
\end{equation}

By definition of $\kappa[I,X]$, if both $H$ and $H'$ do not violate
any constraints, the function $\kappa[I,X]$ if equivalent to $\kappa$,
and equation \ref{eq:monotony-of-constaint-kappa} holds. Furthermore, if $H$
does not violate any constraint but $H'$ does equation
\ref{eq:monotony-of-constaint-kappa} holds trivially:
\begin{align*}
\kappa[I,X](G,\T)=&\kappa[I,X](G,H) \\
\leq& \kappa[I,X](G,H') = \kappa[I,X](G,\T') = \infty
\end{align*}
To complete the proof, it suffices to show that if $H$ violates some constraint,
$H'$ must violate a constraint as well, implying both their costs are $\infty$,
and equation \ref{eq:monotony-of-constaint-kappa} holds.

Notice $\kappa[I,X](G_i,\T_i) = \kappa[I,X](G_i,\induced{H}{\nodes(G_i)})$ and
$\kappa[I,X](G_i,\T'_i) = \kappa[I,X](G_i,\induced{H'}{\nodes(G_i)})$ for
$i=1,2$. We note $H_i = \induced{H}{\nodes(G_i)}$ and $H'_i =
\induced{H'}{\nodes(G_i)}$. $T_i$ and $T'_i$ are clique trees, and from the
junction-tree property $S = \nodes(G_1) \cap \nodes(G_2) \in \minseps(H) \ cap
\minseps(H')$.

Assume $H$ violates some constraint $U \in I \cap X$. If $U
\subseteq \nodes(G_i)$ (for i=1,2), $H_i$ must violate $U$ and:
\[
\kappa[I,X](G_i,H'_i) \geq \kappa[I,X](G_i,H_i) = \infty
\]

We can see some constraint is violated in $H'_i$, and therefore in $H'$. We
deduce $\kappa[I,X](G,H')=\infty$, and equation
\ref{eq:monotony-of-constaint-kappa} holds.

Otherwise, $U$ is not contained in any node of $\T$ or $\T'$ and can not be a
clique in $H$ or $H'$. If $U$ is violated in $H$ then $U \in I$, and $U$ must be
violated in $H'$ as well. 

We can conclude if $H$ violates a constraint, so does $H'$, implying the
lemma holds and the function $\kappa[I,X]$ is split monotone.
\end{proof}

\subsection{Proof of Proposition~\ref{prop:mintraings-propertds}}
\begin{repproposition}{\ref{prop:mintraings-propertds}}
\propmintraingspropertds
\end{repproposition}
\begin{proof}
  Recall from Theorem~\ref{thm:clique-tree-essentials} that the proper
  tree decompositions are precisely the cliques trees of the minimal
  triangulations. It is an easy observation that two minimal
  triangulations have disjoint sets of clique trees (since the two
  have different sets of edges). Since $\kappa$ is a bag cost, the
  cost of each of its clique trees is the same. Hence, we can
  enumerate the proper tree decompositions by increasing cost by
  enumerating the minimal triangulations by increasing cost, and for
  each minimal triangulations we enumerate its clique trees. So, we
  establish the lemma if we can enumerate with polynomial delay all of
  the clique trees of a minimal triangulation. As observed by Carmeli
  et al.~\cite{DBLP:conf/pods/CarmeliKK17}, this can be done by a
  straightforward combination of past results, as explained next.

  Jordan~\cite{Jordan} shows that a tree over the maximal cliques of a
  chordal graph $H$ is a clique tree if and only if it is a maximal
  spanning tree, where the weight of an edge between two maximal
  cliques is the cardinality of their intersection. As the number of
  maximal cliques of a chordal graph is linear in the number of nodes
  (Theorem~\ref{thm:clique-tree-essentials}), this enumeration problem
  is reduced to enumerating all maximal spanning trees, which can be
  solved in polynomial delay~\cite{DBLP:journals/ijcm/YamadaKW10}.
\end{proof}


\section{Case study}\label{app:case}
The experiments in Section~\ref{sec:exper} show that in most cases,
\algname{CKK} has a shorter delay than $\algname{RankedTriang}$,
though opposite situations exist. In this experiment, we focus on two
graphs, one from the Constraint Satisfaction Problem dataset
(\texttt{myciel5g\_3}) and one from the Object Detection dataset
(\texttt{deer\_rescaled\_3020.K15.F1.75}).
Figure~\ref{fig:case-study} shows the runtime and width of the
returned triangulations throughout the execution of the
algorithms. The horizontal axis represents time since the beginning of
the execution, in intervals of $10$ seconds. At each interval, the
number of results, along with the median and minimum widths of these
results, are reported.

Comparing Figures~\ref{fig:csp-ckk} and~\ref{fig:csp-rnk}, we 
observe that \algname{CKK} returns substantially more results for the
CSP graph. On the other hand, the results of \ckk are of higher
width. Our algorithm, \rankt, returns a small amount of results, but they are
all of optimal width. In the case of Figures~\ref{fig:obj-ckk}
and~\ref{fig:obj-rnk}, for the object-detection graph, we can see that
\algname{CKK} has a longer delay. During many intervals, \algname{CKK}
returned a single result. Here too, our algorithm returns only results
of a minimal width. Finally, we observe that in both graphs, the delay
of \rankt is far more stable than that of \algname{CKK}.

{
\begin{figure}[t]
  \centering
  \subfigure[\algname{CKK}\label{fig:csp-ckk}]{\includegraphics[width=0.235\textwidth]{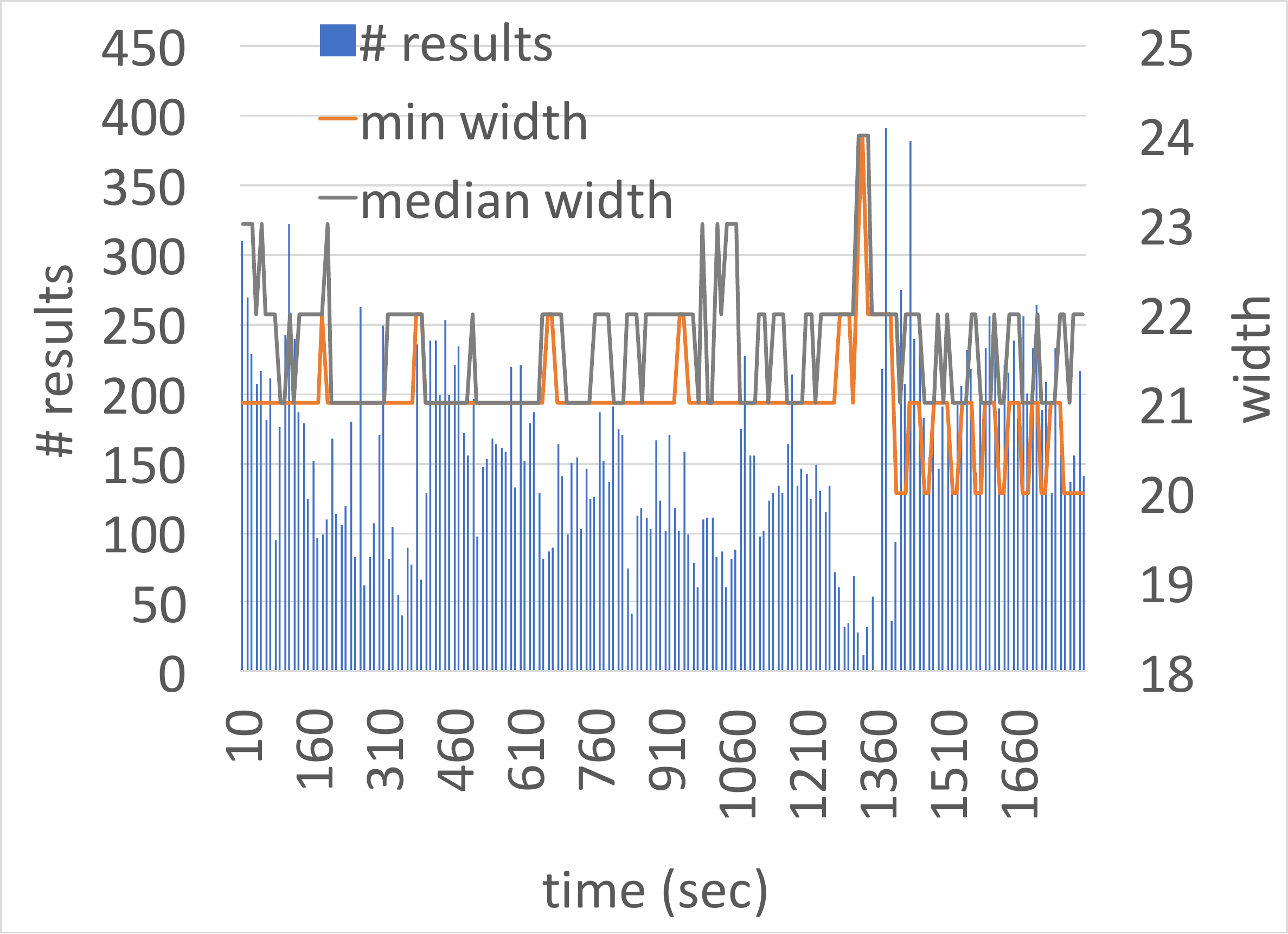}}
  \subfigure[\algname{RankedTriang}\label{fig:csp-rnk}]{\includegraphics[width=0.235\textwidth]{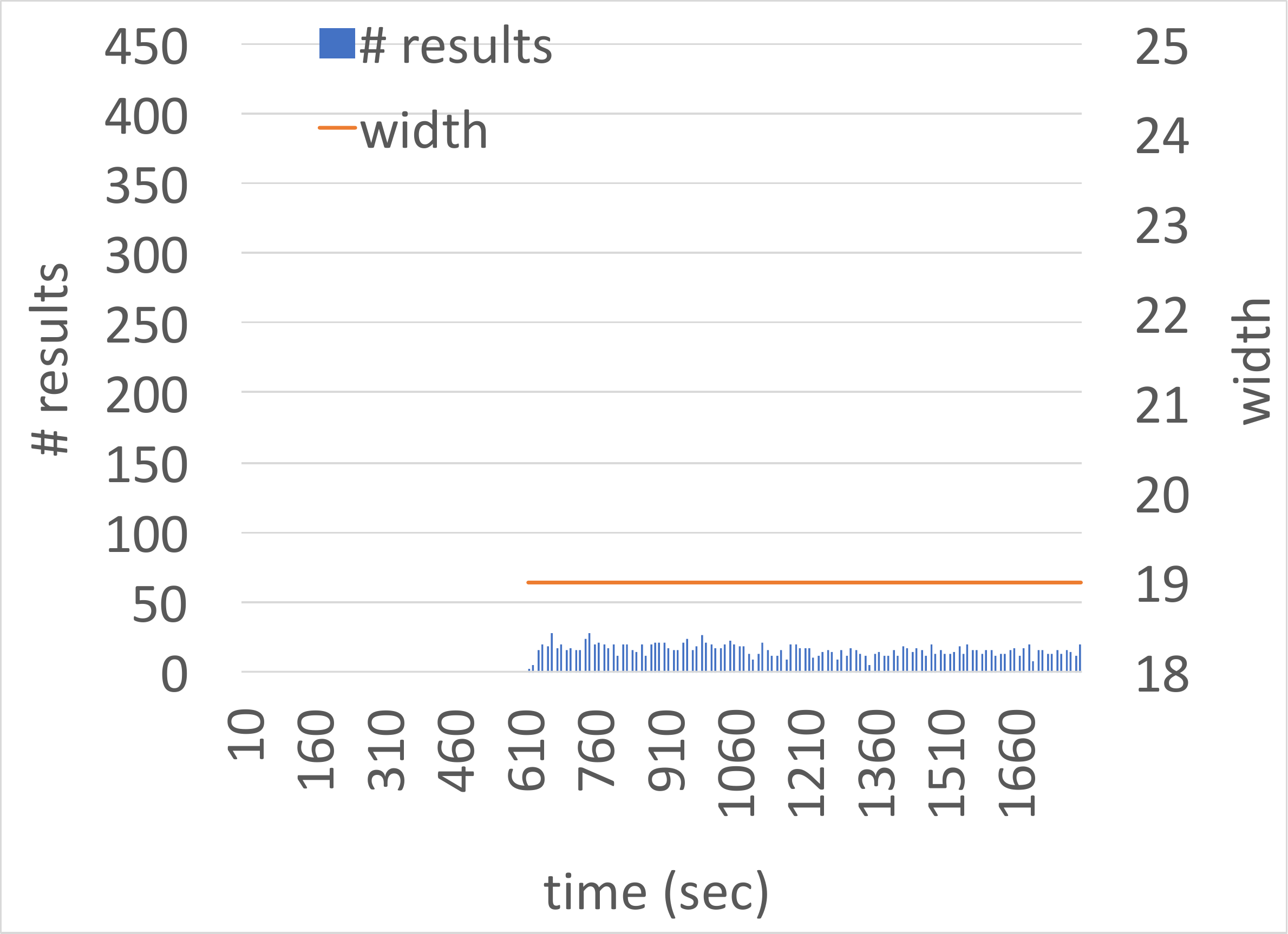}}\\
  \subfigure[\algname{CKK}\label{fig:obj-ckk}]{\includegraphics[width=0.235\textwidth]{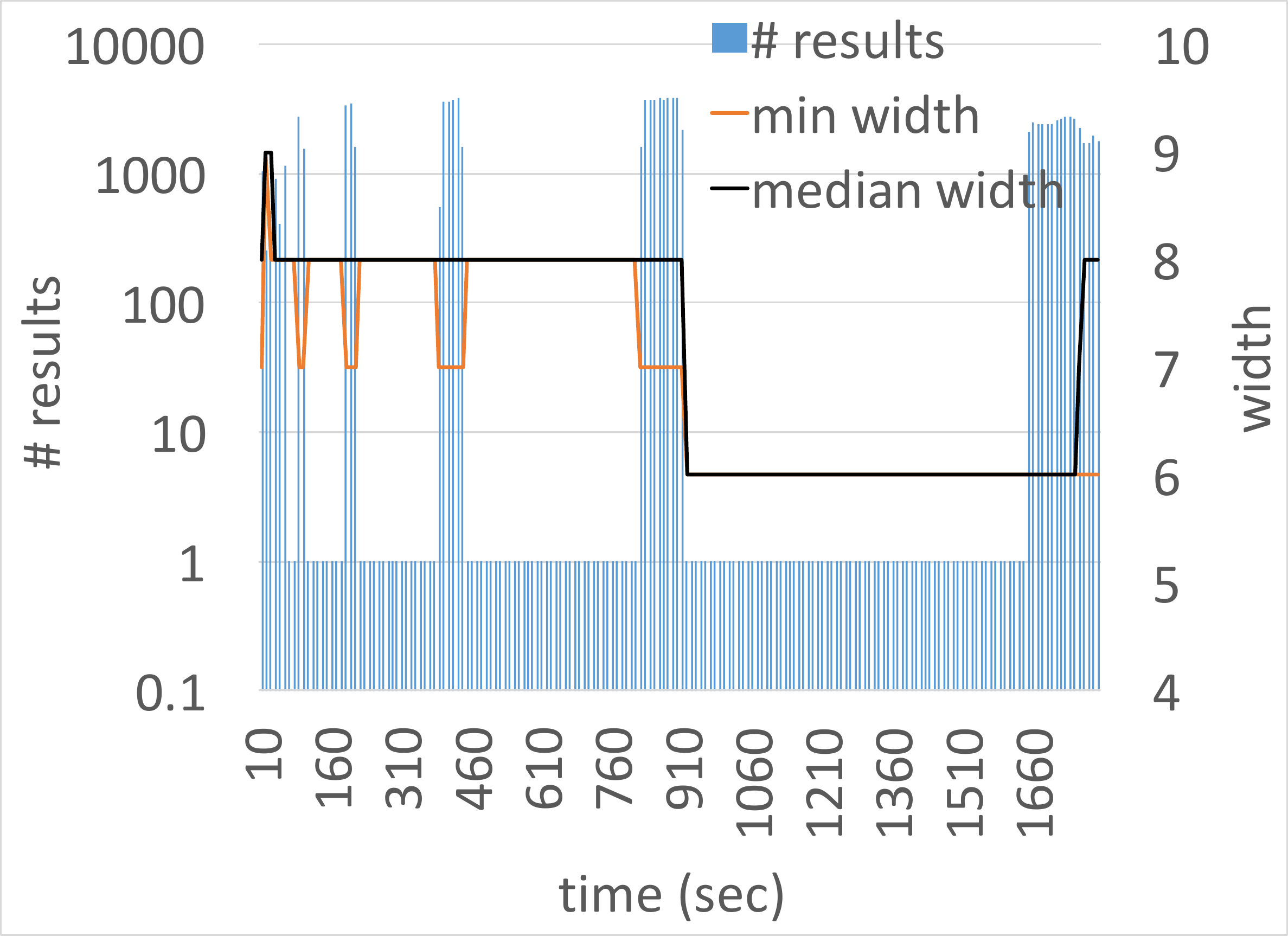}}
  \subfigure[\algname{RankedTriang}\label{fig:obj-rnk}]{\includegraphics[width=0.235\textwidth]{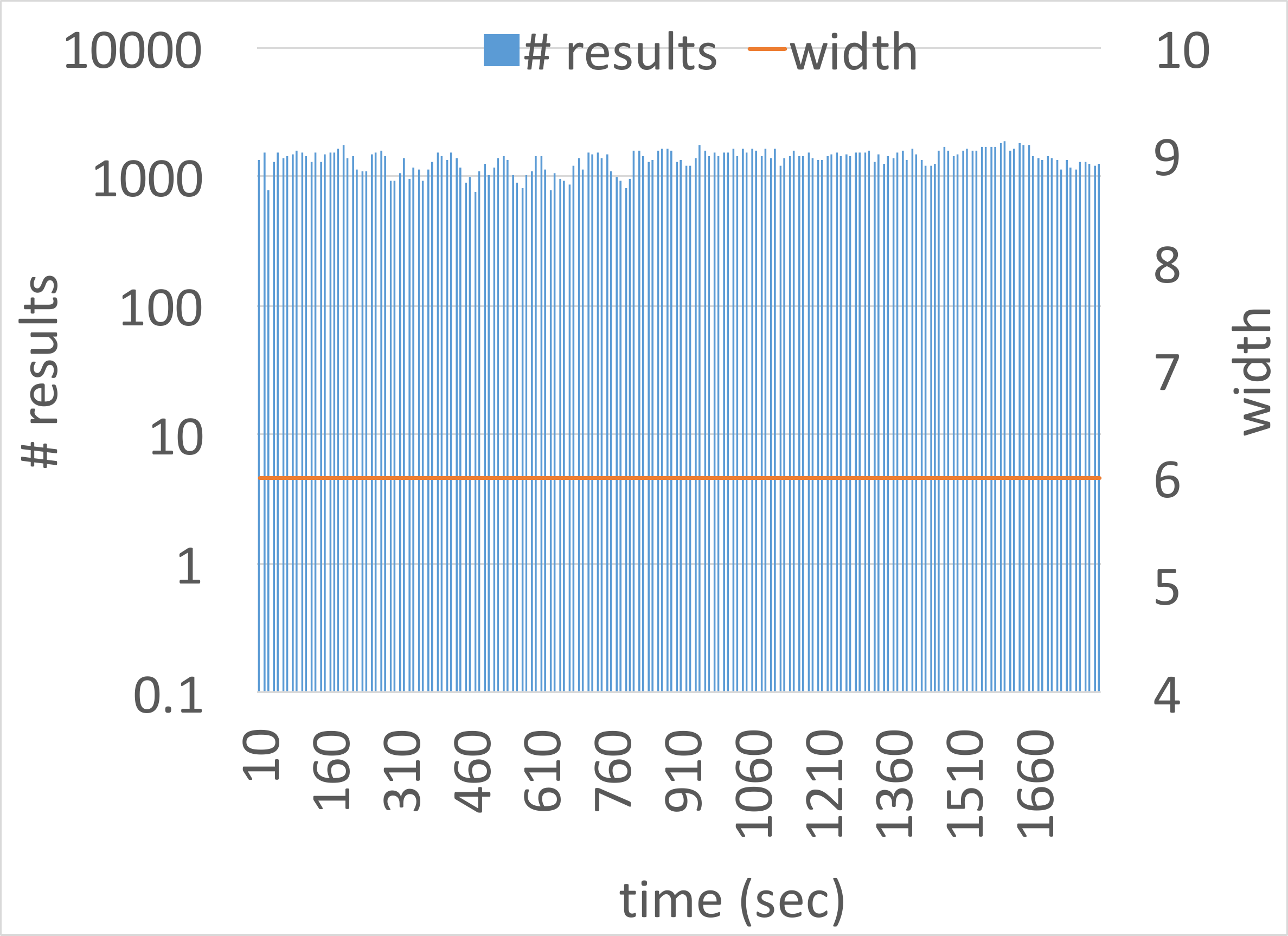}}
  \caption{\label{fig:case-study} Number of triangulations returned
    and their widths on a CSP (top) and object-detection 
    (bottom) graphs.}
\end{figure}
}

\end{document}